\documentclass{amsart}
\usepackage{graphicx} 

\newtheorem{theorem}{Theorem}[section]

\newtheorem{proposition}[theorem]{Proposition}
\newtheorem{corollary}[theorem]{Corollary}
\newtheorem{lemma}[theorem]{Lemma}
\newtheorem{definition}[theorem]{Definition}

\def\eps {{\epsilon}}

\begin{document}

\title[Kinetic Models with Delocalized Collision Integrals]{Local Conservation Laws and Entropy Inequality\\ for Kinetic Models\\ with Delocalized Collision Integrals}

\author{Fr\'ed\'erique Charles}
\address{Universit\'e Grenoble Alpes, LJK, B\^at. IMAG, 150 place du Torrent, 38400 Saint Martin d'Hères, France}
\email{frederique.charles@univ-grenoble-alpes.fr}

\author{Zhe Chen}
\address{Sorbonne Universit\'e, LJLL, Bo\^\i te courrier 187, 75252 Paris Cedex 05 France}
\email{zhe.chen@sorbonne-universite.fr}

\author{Fran\c cois Golse}
\address{\'Ecole polytechnique \& IP Paris, CMLS, 91128 Palaiseau Cedex, France}
\email{francois.golse@polytechnique.edu}

\begin{abstract}
This article presents a common setting for the collision integrals $\mathrm{St}$ appearing in the kinetic theory of dense gases. It includes the collision integrals of the Enskog equation, of (a variant of) the 
Povzner equation, and of a model for soft sphere collisions proposed by Cercignani [Comm. Pure Appl. Math. \textbf{36} (1983), 479--494]. All these collision integrals are ``delocalized'', in the sense that 
they involve products of the distribution functions of gas molecules evaluated at positions whose distance is of the order of the molecular radius. Our first main result is to express these collision integrals as
the divergence in $v$ of some mass current, where $v$ is the velocity variable, while $v_i\mathrm{St}$ and $|v|^2\mathrm{St}$ are expressed as the phase space divergence (i.e divergence in both position 
and velocity) of appropriate momentum and energy currents. This extends to the case of dense gases an earlier result by Villani [Math. Modelling Numer. Anal. M2AN \textbf{33} (1999), 209--227] in the
case of the classical Boltzmann equation (where the collision integral is involves products of the distribution function of gas molecules evaluated at different velocities, but at the same position. Applications
of this conservative formulation of delocalized collision integrals include the possibility of obtaining the local conservation laws of momentum and energy starting from this kinetic theory of denses gases.
Similarly a local variant of the Boltzmann H Theorem, involving some kind of free energy instead of Boltzmann's H function, can be obtained in the form of an expression for the entropy production in terms 
of the phase space divergence of some phase space current, and of a nonpositive term.
\end{abstract}



\maketitle


\section*{Introduction}


We are concerned in this paper with the basic structure of some kinetic equations used to model non perfect monatomic gases. The classical kinetic theory of gases of Maxwell (1866) and Boltzmann (1872) 
is limited to perfect monatomic gases. It has been later generalized, in particular by Enskog (1922), to treat the case of dense gases (see chapter 16 of \cite{CC}, for instance). Further generalizations of 
Enskog's theory --- most notably Povzner's kinetic equation (1962), and Cercignani's model for the kinetic theory of soft spheres (1980) --- are presented in section \ref{sec:intro}.

The classical kinetic theory of gases is intimately related to the Euler equations of (inviscid) gas dynamics through the local conservation laws of mass, momentum and energy, satisfied by classical solutions 
of the Boltzmann equation (see sections 3.1 and 3.3 in \cite{CIP}, and the remarks in section \ref{sec:LocConsLaws}). At the time of this writing, it would seem that these \textit{local} conservation laws of mass, 
momentum and energy are not known to be verified by classical solutions of the Enskog equation or its variants presented in section \ref{sec:intro}.

For this reason, we propose in this work a systematic method for obtaining local conservation laws in the theory of kinetic models of Enskog type. Our arguments are based on writing the collision integrals
of Enskog type as the divergence in the position and velocity variables of mass, momentum and energy currents, following a little-known idea of Landau on the Boltzmann collision integral (see for instance
\S 41 in \cite{LL10} and \cite{VillaniLandauCurr}).

In the case of the Boltzmann equation, the Friedrichs-Lax entropy condition used in the theory of the Euler system of (inviscid) gas dynamics (see sections 3.2, 4.5 and chapter V of \cite{Dafermos}) can be 
deduced from the local form of the Boltzmann H theorem \cite{BardosFG}. For this reason, we also discuss the possibility of obtaining \textit{local} entropy inequalities for classical solutions of Enskog-type 
kinetic equations.

The outline of this paper is as follows: various examples of kinetic models with delocalized collision integrals analogous to Enskog's are presented in section \ref{sec:intro}. Local conservation laws and
its applications to fluid dynamics are briefly recalled in section \ref{sec:LocConsLaws}. Our main results are presented in section \ref{sec:MainR}. In particular, a general formulation of delocalized collision
integrals is presented in section \ref{ssec:GenDelocCollInt}, and the corresponding global conservation laws of mass, momentum and energy in \ref{ssec:GlobConsDelocCollInt} are briefly recalled for
these generalizations of Enskog's equation. The core of this paper is Theorem \ref{T-Currents}, in section \ref{ssec:LandauCurrents}, providing an analogue of Landau's representation of Boltzmann's
collision integral as the divergence in $v$ of a mass current, in the more complex case of delocalized collision integral. Our second main result is the local entropy inequality for a subclass of delocalized
collision integrals presented in Theorem \ref{T-HThmLoc} of section \ref{ssec:Entropy}. Applications of our results to the continuum mechanics of dense gases are discussed in section \ref{sec:Appli}.


\section{Examples of Delocalized Collision Integrals\\ in the Kinetic Theory of Gases}\label{sec:intro}


All collisional kinetic models take the form
\begin{equation}\label{GenKinMod}
(\partial_t+v\cdot\nabla_x)f(t,x,v)=\mathrm{St}[f](t,x,v)\,,
\end{equation}
where $f\equiv f(t,x,v)\ge 0$ is the velocity distribution function, i.e. the number density of gas molecules at position $x\in\mathbf R^3$ and with velocity $v\in\mathbf R^3$ at time $t$, while $\mathrm{St}$ 
is the collision operator, usually a nonlinear integral operator acting on the velocity, and sometimes the position variables in $f$.

In the classical kinetic theory of gases due to Maxwell (1866) and Boltzmann (1872), the collision integral computes the variation in the number of gas molecules with velocity $v\in\mathbf R^3$ due to binary,
elastic collisions with gas molecules at the same position with a velocity different from $v$. Specifically, if the gas molecules are hard spheres of radius $r>0$, the Boltzmann collision integral is\footnote{See
in \cite{Sone2} the paragraph following equations (1.7) in chapter 1, especially footnote 6 on p. 4 in that chapter.}
\begin{equation}\label{BoltzCollInt}
\mathcal B[f,f](v):=\iint_{\mathbf R^3\times\mathbf S^2}(f(v')f(w')-f(v)f(w))(2r)^2(v-w|n)_+dwdn\,,
\end{equation}
where $f\ge 0$ is the velocity distribution function, while $(u|v)$ is the inner product of $u,v\in\mathbf R^3$, and $z_+:=\max(z,0)$ for all $z\in\mathbf R$ (see section 2.5 in \cite{CIP}). In the integrand above, 
\begin{equation}\label{CollTransfo}
v'\equiv v'(v,w,n):=v-(v-w|n)n\,,\quad w'\equiv w'(v,w,n):=w+(v-w|n)n\,.
\end{equation}
Notice that these transformations are invariant under the substitution $n\mapsto -n$.

One easily checks that $v'(v,w,n),w'(v,w,n)$ runs through the set of all the solutions of the system of equations
\begin{equation}\label{MicroCons}
v'+w'=v+w\,,\quad |v'|^2+|w'|^2=|v|^2+|w|^2
\end{equation}
as $n$ runs through the unit sphere $\mathbf S^2$ of $\mathbf R^3$. The first equality above is the conservation of momentum for a pair of point particles with mass $m>0$, while the second equality is the 
conservation of kinetic energy for the same particles in any elastic collision involving such particles located at the same position in $\mathbf R^3$, with velocities $v,w$ immediately after the collision, and
$v',w'$ immediately before the collision. Without entering the details of the derivation of the Boltzmann collision integral, suffices it to say that, in that integral, one should think of the term $-f(v)f(w)$
as measuring the loss of gas molecules with velocity $v$ due to an impending collision with another gas molecule at the same position with velocity $w$. Similarly, one should think of the term $f(v')f(w')$
as measuring the creation of gas molecules with velocity $v$ resulting from the collision of a gas molecule with velocity $v'$ with a gas molecule with velocity $w'$ immediately before the collision. Thus
all the pair of velocities involved in the Boltzmann collision integrand are always velocities of gas molecules \textit{about to collide} (and never of gas molecules having just collided). This ``justifies'' the 
assumption of Maxwell and Boltzmann that such particle velocities are statistically independent, so that their joint distribution is $f(v')f(w')$ in the latter case, and $f(v)f(w)$ in the former.

Therefore, in the Boltzmann equation for the hard sphere gas, the molecular radius appears only as a scaling factor of the collision differential cross-section $(2r)^2(\cos(\widehat{v-w,n}))_+$, and never in
the distribution function itself in the collision integral. This has been justified (for instance by Lanford \cite{Lanford75}, see also chapters 2 and 4 in \cite{CIP} for a streamlined presentation of Lanford's 
arguments) in the regime known as the ``Boltzmann-Grad limit'' of the classical mechanics of $N$-particle system, where $N\gg 1$ while $r\ll 1$ satisfy $Nr^2\sim 1$. In particular $Nr^3=O(r)\to 0$ in that 
limit, so that there is no excluded volume, and the resulting equations of state for the pressure and the internal energy correspond to a perfect gas.

Situations where the molecular radius is larger may fail to satisfy the Boltzmann-Grad assumption. Alternatively, one might also seek corrections in $r$ to the perfect gas equation of state for the pressure 
and the internal energy. Whatever the motivation, one might be led to consider collision integrals involving a \textit{delocalized} collision integrand, i.e. a collision integrand where gas molecules are not 
considered as point particles located at the same point, but are hard spheres of radius $r>0$.

\begin{figure}
\includegraphics[width=6cm]{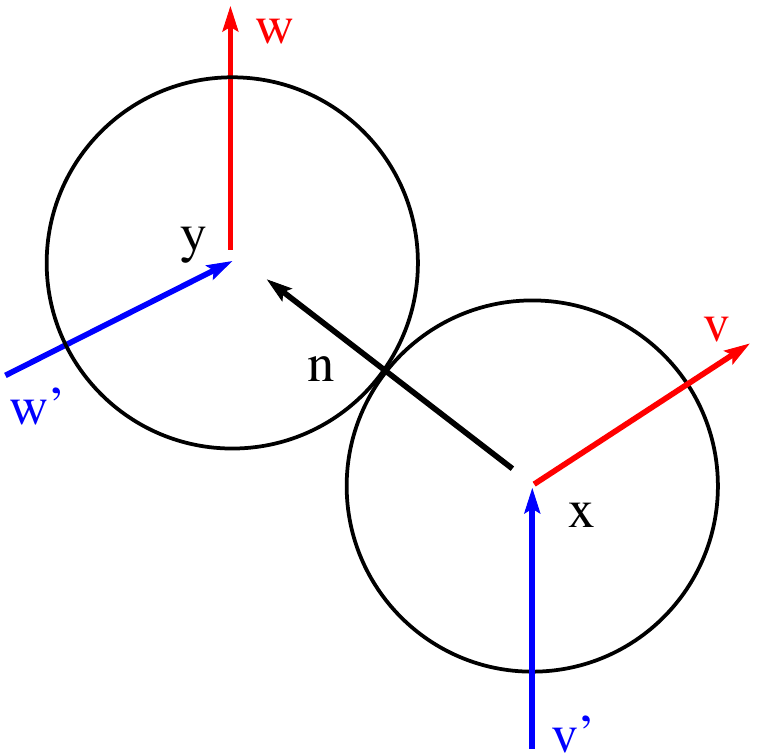}
\caption{Two hard spheres centered at $x$ and $y=x+rn$, with pre-collision velocities $v'$ and $w'$, and post-collision velocities $v$ and $w$. Observe that the post-collision velocities $v$ and $w$ satisfy
the condition $(w-v|n)=(w-v|\tfrac{y-x}{|y-x|})>0$, while the pre-collision velocities satisfy $(w'-v'|n)=(w'-v'|\tfrac{y-x}{|y-x|})<0$.}
\end{figure}

Perhaps the best known kinetic model with such delocalized collision integral is the Enskog-Boltzmann equation (see \cite{ArkerydEnskSIMA,ToscaBello}). In the hard sphere case, the Enskog-Boltzmann 
collision integral takes the form
\begin{equation}\label{EnskBCollInt}
\begin{aligned}
\mathcal E[f,f](x,v):=\iint_{\mathbf R^3\times\mathbf S^2}(&f(x,v'(v,w,n))f(x-2rn,w'(v,w,n))
\\
&-f(x,v)f(x+2rn,w))4r^2(v-w|n)_+dwdn\;.
\end{aligned}
\end{equation}
This collision integral can be recast as
\[
\begin{aligned}
\mathcal E[f,f](x,v):=\iint_{(\mathbf R^3)^2}(&f(x,v'(v,w,n))f(x-z,w'(v,w,n))
\\
&-f(x,v)f(x+z,w))\big(v-w|\tfrac{z}{|z|}\big)_+\delta_{2r}(|z|)dwdz\;,
\end{aligned}
\]
where $\delta_{2r}$ is the Dirac ``function'' at $2r>0$ on $\mathbf R$.

A first obvious generalization is based on the last formula for the Enskog-Boltzmann collision integral. Pick $\zeta\in C(\mathbf R)$ satisfying
\[
\zeta\ge 0\,,\qquad\mathrm{supp}(\zeta)=[-1,1]\,,\qquad\int_\mathbf R\zeta(s)ds=1\;,
\]
and set 
\[
\zeta_\eps(s):=\tfrac1\eps\zeta\left(\tfrac{s}\eps\right)\,,\qquad\theta_{\eps,r}(|z|)=\zeta_\eps(|z|-2r)\;,
\]
so that $\theta_{\eps,r}\to \delta_{2r}$ in $\mathcal D'(\mathbf R)$ as $\eps\to 0^+$. If one replaces the Dirac measure $\delta_{2r}(|z|)$ with its regularized variant $\theta_{\eps,r}(|z|)$,
one arrives at the expression
\begin{equation}\label{SoftSphCollInt}
\begin{aligned}
\mathcal C[f,f](x,v):=\iint_{(\mathbf R^3)^2}\Big(&f(x,v'(v,w,\tfrac{z}{|z|}))f(x-z,w'(v,w,\tfrac{z}{|z|}))
\\
&-f(x,v)f(x+z,w)\Big)\theta_{\eps,r}(|z|)\left(v-w|\tfrac{z}{|z|}\right)_+dwdz\;.
\end{aligned}
\end{equation}
One can think of this as a collision integral for ``soft'' spheres --- soft meaning that the radius of the colliding balls is slightly distributed around $r>0$, instead of being exactly equal to $r$; otherwise 
the collision transformation of $(v,v_*)\mapsto(v',v'_*)$ is the same as above. In other words, these collisions are elastic. (It is unlikely that collisions between pairs of soft spherical particles really
behave in this manner, which is why the word soft is put between quotes.) See \cite{CerciCPAM} for a presentation of this model, which is derived from a large particle system in the style of Lanford's
argument for the Boltzmann equation. In particular, Cercignani proposed a specific formula for $\theta_{\eps,r}$: see equations (2.1)-(2.5) on p. 481 in \cite{CerciCPAM}.

The Cercignani soft-sphere collision integral can be recast as
\[
\begin{aligned}
\mathcal C[f,f](x,v)\!:=&\iint_{(\mathbf R^3)^2}\!\!\Big(f(x,v'(v,w,\tfrac{x-y}{|x-y|}))f(y,w'(v,w,\tfrac{x-y}{|x-y|}))\mathbf 1_{(v-w|x-y)>0}
\\
&-\!f(x,v)f(y,w)\mathbf 1_{(v-w|x-y)<0}\Big)\theta_{\eps,r}(|x-y|)\left|\left(v\!-\!w|\tfrac{x-y}{|x-y|}\right)\right|dwdy\,.
\end{aligned}
\]
A last generalization consists in forgetting the specifics of the distribution $\theta_{\eps,r}$. This suggests considering instead the Povzner type collision integral
\begin{equation}\label{PovznerCollInt}
\begin{aligned}
\mathcal P[f,f](x,v)\!:=\!\iint_{\mathbf R^3\times\mathbf R^3}\Big(f(x,v'(v,w,\tfrac{x-y}{|x-y|}))f(y,w'(v,w,\tfrac{x-y}{|x-y|}))\mathbf 1_{(v-w|x-y)>0}
\\
-f(x,v)f(y,w)\mathbf 1_{(v-w|x-y)<0}\Big)b(|x-y|)\left|\left(v-w|\tfrac{x-y}{|x-y|}\right)\right|dwdy
\end{aligned}
\end{equation}
where $b\in C_c([0,+\infty))$ is a nonnegative function. See \cite{Povzner}, where this model\footnote{Povzner's original collision integral in \cite{Povzner} is
\[
\begin{aligned}
\tilde{\mathcal P}[f,f](x,v):=\iint_{(\mathbf R^3)^2}\Big(&f(x,v'(v,w,\tfrac{x-y}{|x-y|}))f(y,w'(v,w,\tfrac{x-y}{|x-y|}))
\\
&-f(x,v)f(y,w)\Big)b(|x-y|)\left|\left(v-w|\tfrac{x-y}{|x-y|}\right)\right|dwdy\,.
\end{aligned}
\]
Since the transformation $(v,w)\mapsto(v'(v,w,n),w'(v,w,n))$ is invariant under the central symmetry $n\mapsto-n$, 
\[
\mathcal B[f,f](v)=\iint_{\mathbf R^3\times\mathbf S^2}(f(v')f(w')-f(v)f(w))2r^2|(v-w|n)|dwdn\,.
\]
Povzner's original definition of his delocalized collision integral is based on this last form of the Boltzmann integral instead of \eqref{BoltzCollInt}. Since $(v'-w'|x-y)=-(v-w|x-y)<0$ for the gain 
part $f(x,v')f(y,w')$ and $(v-w|x-y)<0$ for the loss part $f(x,v)f(y,w)$ of the collision integral $\mathcal P$, only particles about to collide are uncorrelated in \eqref{PovznerCollInt}, as in the 
case of \eqref{BoltzCollInt}. That this important feature of Boltzmann's theory is absent from \cite{Povzner} may be an oversight. For this reason, we prefer considering $\mathcal P$ instead 
of $\tilde{\mathcal P}$.} is introduced for the first time.

Along with the Enskog-Boltzmann collision integral above, one also finds in the literature the genuine Enskog equation (see chapter 16 of \cite{CC} or \cite{BelloLacho1}) and a modified, 
or revised Enskog collision integral (see \cite{Resibois}, or \cite{vanBeijeren} and \cite{vBErnst1,vBErnst2}):
\begin{equation}\label{EnskCollInt}
\begin{aligned}
\tilde{\mathcal E}[f,f](x,v)\!\!:=\!\!\iint_{\mathbf R^3\times\mathbf S^2}\!\!\big(Y(\rho_f(x),\rho_f(x\!-\!2rn))f(x,v'(v,w,n))f(x\!-\!2rn,w'(v,w,n)))
\\
-\!Y(\rho_f(x),\rho_f(x\!+\!2rn))f(x,v)f(x\!+\!2rn,w)\big)4r^2(v\!-\!w|n)_+dwdn&,
\end{aligned}
\end{equation}
where the factor $Y:\,[0,+\infty)^2\to[0+\infty)$ is used to enhance the collision frequency, and where
\begin{equation}\label{Def-rho}
\rho_f(x):=\int_{\mathbf R^3}f(x,v)dv\,.
\end{equation}


\section{Local Conservation Laws for the Boltzmann Collision Integral}\label{sec:LocConsLaws}


Returning to the Boltzmann collision integral \eqref{BoltzCollInt}, it is well known that, for all $f\in L^1(\mathbf R^3,(1+|v|)^3dv)$, 
\begin{equation}\label{ConsLawBoltzCollInt}
\int_{\mathbf R^3}\mathcal B[f,f](v)\left(\begin{matrix}1\\ v_1\\ v_2\\ v_3\\ |v|^2\end{matrix}\right)dv=\left(\begin{matrix}0\\ 0\\ 0\\ 0\\ 0\end{matrix}\right)\,,
\end{equation}
see for instance sections 3.1 and 3.3 in \cite{CIP}. These five identities correspond to the local conservation of mass (or particle number), of the three components of momentum, and of energy in the collision 
process described by the Boltzmann collision integral.

As a consequence of these identities, if $f$ is a classical solution of the Boltzmann equation 
\[
(\partial_t+v\cdot\nabla_x)f(t,x,v)=\mathcal B[f,f](t,x,v)(=\mathcal B[f(t,x,\cdot),f(t,x,\cdot)](v))
\]
such that $f(t,x,\cdot),\partial_tf(t,x,\cdot)$ and $\nabla_xf(t,x,\cdot)\in L^1(\mathbf R^3,(1+|v|)^3dv)$, then
\begin{equation}\label{ConsLawBoltz}
\partial_t\int_{\mathbf R^3}f(t,x,v)\left(\begin{matrix}1\\ v_1\\ v_2\\ v_3\\ \frac12|v|^2\end{matrix}\right)dv+\nabla_x\cdot\int_{\mathbf R^3}vf(t,x,v)\left(\begin{matrix}1\\ v_1\\ v_2\\ v_3\\ \frac12|v|^2\end{matrix}\right)dv=0\,.
\end{equation}
These five differential identities are of considerable importance in connecting the kinetic description of gases with continuum mechanics. Indeed, if one calls
\begin{equation}\label{DefRhoU}
\rho(t,x):=\int_{\mathbf R^3}f(t,x,v)dv\ge 0\,,\qquad u(t,x)=\frac{\mathbf 1_{\rho(t,x)>0}}{\rho(t,x)}\int_{\mathbf R^3}vf(t,x,v)dv\in\mathbf R^3\,,
\end{equation}
the first identity above takes the form of the equation of continuity of fluid mechanics
\begin{equation}\label{ContEq}
\partial_t\rho(t,x)+\nabla_x\cdot(\rho(t,x)u(t,x))=0\,,
\end{equation}
while the second identity becomes
\begin{equation}\label{EulEqBoltz}
\partial_t(\rho(t,x)u(t,x))+\nabla_x\cdot(\rho(t,x)u(t,x)^{\otimes 2})+\nabla_x\cdot P(t,x)=0\,,
\end{equation}
where 
\begin{equation}\label{DefP}
P(t,x):=\int_{\mathbf R^3}(v-u(t,x))^{\otimes 2}f(t,x,v)dv\,.
\end{equation}
Finally, the third identity takes the form
\begin{equation}\label{EnergEqEulBoltz}
\begin{aligned}
\partial_t(\tfrac12\rho(t,x)|u(t,x)|^2+\tfrac12\mathrm{tr}P(t,x))+\nabla_x\cdot(u(t,x)(\tfrac12\rho(t,x)|u(t,x)|^2+\tfrac12\mathrm{tr}P(t,x)))
\\
+\nabla_x\cdot(P(t,x)\cdot u(t,x)+Q(t,x))=0&\,,
\end{aligned}
\end{equation}
where $Q$ is the time-dependent vector field defined by
\begin{equation}\label{DefQ}
Q(t,x):=\tfrac12\int_{\mathbf R^3}(v-u(t,x))|v-u(t,x)|^2f(t,x,v)dv\,.
\end{equation}
If $f(t,x,v)$ is an even function of $v-u(t,x)$ one easily checks (by a parity argument) that 
\[
Q(t,x)=0\,,
\]
These identities are obviously reminiscent of the Euler system of equations for compressible fluids (see (3.3.17) in chapter III of \cite{Dafermos}), and the tensor field $P(t,x)$ plays the role of 
the pressure. Specifically if $f$ is a radial function of $v-u(t,x)$  --- for instance if $f$ is a local Maxwellian equilibrium in $v$ centered at $u(t,x)$ --- one finds that
\[
P(t,x)=p(t,x)I\,,\quad\text{ and hence }\quad\mathrm{tr}P(t,x)=3p(t,x)I\,.
\]
Hence
\begin{equation}\label{EulEq}
\partial_t(\rho(t,x)u(t,x))+\nabla_x\cdot(\rho(t,x)u(t,x)^{\otimes 2})+\nabla_xp(t,x)=0\,,
\end{equation}
while
\begin{equation}\label{EnergEq}
\partial_t(\tfrac12\rho(t,x)|u(t,x)|^2+\tfrac32p(t,x))+\nabla_x\cdot(u(t,x)(\tfrac12\rho(t,x)|u(t,x)|^2+\tfrac52p(t,x)))=0\,.
\end{equation}
Setting $\theta(t,x)=p(t,x)/\rho(t,x)$, one finds that these are the Euler equation and local energy conservation law for a perfect gas with adiabatic index 
\[
\gamma=(5/2)/(3/2)=5/3\,.
\]
It is well known that the adiabatic index of a perfect gas is $\gamma=1+2/n$, where $n$ is the number of degrees of freedom of the gas molecule (see formulas (44.1-2) in \cite{LL5}). For a 
monatomic gas, each molecule has $n=3$ degrees of freedom (corresponding to translations in the directions of the coordinate axis), so that $\gamma=5/3$ is the adiabatic index of a perfect
monatomic gas. Thus the collision integral \eqref{BoltzCollInt} can describe only perfect monatomic gases.
 
\smallskip
All these considerations suggest studying the following question: compute
\[
\int_{\mathbf R^3}\mathrm{St}[f](t,x,v)\left(\begin{matrix}1\\ v_1\\ v_2\\ v_3\\ |v|^2\end{matrix}\right)dv\,,
\]
where $\mathrm{St}[f]$ is any one of the delocalized collision integrals presented above, i.e. $\mathcal E[f,f]$ (Enskog-Boltzmann) or $\tilde{\mathcal E}[f,f]$ (revised Enskog), or $\mathcal C[f,f]$ 
(Cercignani's soft-sphere collision integral), or $\mathcal P[f,f]$ (Povzner-type collision integral). 

In general, delocalization may prevent some of these quantities to be identically $0$. However, it is certainly interesting to study these expressions in view of applications to fluid mechanics.


\section{Main Results: Delocalized Collision Integrals\\ in Conservative Form and Applications}\label{sec:MainR}


In order to approach the question of local conservation laws for delocalized collision integrals, we shall make a small detour. While it may seem unnecessary, it sheds some light on the structure
of these collision integrals.

Our approach finds its origin in a seemingly little-known property of the Boltzmann collision integral. The fact that the Boltzmann collision integral satisfies the local conservation of mass suggests 
seeking a mass current $J[f,f](v)$ such that
\[
\mathcal B[f,f](v)=-\nabla_v\cdot J[f,f](v)
\]
for all $f\in\mathcal S(\mathbf R^3)$ (the Schwartz space of smooth functions with rapidly decaying derivatives of all orders). This is the starting point of the derivation of the Landau collision integral 
from the Boltzmann collision integral in the grazing collisions regime (see \S41 in \cite{LL10}). A more systematic (and rigorous) discussion of this property can be found in \cite{VillaniLandauCurr}.

\subsection{Delocalized Collision Integrals: a General Setting}\label{ssec:GenDelocCollInt}


In this section, we propose a common setting for the delocalized collision integrals $\mathcal E,\mathcal C,\mathcal P,\tilde{\mathcal E}$ presented in section \ref{sec:intro}.

The general form of delocalized collision integrals considered in this paper is
\begin{equation}\label{GenCollInt}
\begin{aligned}
\mathrm{St}[f,g](x,v):=\int_{(\mathbf R^3)^2}(&f(x,v'(v,w,n_{x,y}))g(y,w'(v,w,n_{x,y}))\mathbf 1_{(v-w|n_{x,y})>0}
\\
&-f(x,v)g(y,w)\mathbf 1_{(v-w|n_{x,y})<0})B[\rho_f,\rho_g](x,y,v-w)dydw\,,
\end{aligned}
\end{equation}
where
\[
\rho_f(x):=\int_{\mathbf R^3}f(x,v)dv\,,\qquad\rho_g(y):=\int_{\mathbf R^3}g(y,w)dw
\]
are the macroscopic densities of the distribution functions $f$ and $g$, while
\[
n_{x,y}:=\frac{x-y}{|x-y|}\,,\qquad x\not=y\,.
\]

The collision kernel $B[\rho_f,\rho_g](x,y,v-w)\ge 0$ is in general a functional of the macroscopic densities $\rho_f$ and $\rho_g$, and a function --- or more generally a distribution --- of the positions 
$x,y$, and of the relative velocity $v-w$. In addition, it satisfies
\[
\left\{
\begin{aligned}
{}&(i) &&B[\varpi,\rho](y,x,w-v)=B[\rho,\varpi](x,y,v-w)\,,
\\
&(ii) &&B[\rho,\varpi](x,y,v-w)=B[\rho,\varpi](x,y,v'(v,w,n_{x,y})-w'(v,w,n_{x,y}))\,,
\\
&(iii) &&B[\rho,\varpi](x,y,v-w)=0\text{ if }|x-y|>R\,,
\end{aligned}
\right.
\]
for some $R>0$, a real parameter which can be thought of as being the radius of the sphere of influence of a gas molecule\footnote{In other words, $R$ is the distance beyond which a point particles 
is not affected by the presence of one gas molecule. Assuming short range molecular interactions, one can think of $R$ as being of the order of a few van der Waals radii of the gas molecules.} of the 
gas molecules, together with the bound
\[
\begin{aligned}
{}&(iv) &&\sup_{v,w\in\mathbf R^3}\frac1{1+|v-w|}\int_{(\mathbf R^3)^2}B[\rho,\varpi](x,y,v-w)\mathbf 1_{|x|\le A}dxdy<\infty\,.
\end{aligned}
\]

Let us check that the examples of delocalized collision integrals considered in section \ref{sec:intro} can be put in the form \eqref{GenCollInt}.

\noindent
\underline{(1) Enskog-Boltzmann equation:} in that case
\[
B[\rho,\varpi](x,y,v-w)=(v-w|n_{x,y})|\delta_{2r}(|x-y|)\,,
\]
(and of course, the collision kernel is independent of the macroscopic densities $\rho,\varpi$);

\noindent
\underline{(2) Soft spheres:} in that case
\[
B[\rho,\varpi](x,y,v-w)=\theta_{\eps,r}(|x-y|)|(v-w|n_{x,y})|
\]
(and of course, the collision kernel is independent of the macroscopic densities $\rho,\varpi$);

\noindent
\underline{(3) Povzner-type equation:} in that case
\[
B[\rho,\varpi](x,y,v-w)=b(|x-y|)|(v-w|n_{x,y})|
\]
(and of course, the collision kernel is independent of the macroscopic densities $\rho,\varpi$). This type of collision kernel will be used below in section \ref{ssec:Entropy}  in connection with the H Theorem;

\noindent
\underline{(4) Enskog equation:} in that case
\[
B[\rho,\varpi](x,y,v-w)=Y[\rho(x),\varpi(y)]|(v-w|n_{x,y})|\delta_{2r}(|x-y|)\,.
\]

\subsection{Global Conservation Laws for Delocalized Collision Integrals}\label{ssec:GlobConsDelocCollInt}


The following result is well-known in the case of the Enskog equation \cite{Polewczak,EstebanPerth}. We briefly sketch its proof in the more general setting considered here.

\begin{proposition}\label{P-GlobCons}
For all measurable, rapidly decaying $f,g$ defined a.e. on  $\mathbf R^3\times\mathbf R^3$, the collision integral $\mathrm{St}$ satisfies the local mass conservation identity
\[
\int_{\mathbf R^3}\mathrm{St}[f,g](x,v)dv=0\,.
\]
It also satisfies the global momentum and energy conservation identities
\[
\int_{\mathbf R^3\times\mathbf R^3}\mathrm{St}[f,f](x,v)v_jdxdv=\int_{\mathbf R^3\times\mathbf R^3}\mathrm{St}[f,f](x,v)|v|^2dxdv=0\,,\quad j=1,2,3\,.
\]
\end{proposition}

\begin{proof}
For all measurable, rapidly decaying $f,g$ defined a.e. on $\mathbf R^3\times\mathbf R^3$,
\[
\begin{aligned}
\int_{\mathbf R^3}\mathrm{St}[f,g](x,v)dv=\int_{(\mathbf R^3)^3}\big(f(x,v'(v,w,n))g(y,w'(v,w,n))\mathbf 1_{(v-w|n_{x,y})>0}
\\
-f(x,v)f(y,w)\mathbf 1_{(v-w|n_{x,y})<0}\big)B[\rho_f,\rho_g](x,y,v-w)dvdwdy
\\
=\int_{(\mathbf R^3)^3}\big(B[\rho_f,\rho_g](x,y,v'(v,w,n_{x,y})-w'(v,w,n_{x,y}))\mathbf 1_{(v'(v,w,n_{x,y})-w'(v,w,n_{x,y}|n_{x,y})>0}
\\
-B[\rho_f,\rho_g](x,y,v-w)\mathbf 1_{(v-w|n_{x,y})<0}\big)f(x,v)g(y,w)dvdwdy
\\
\text{(with the change of variables $(v,w)\mapsto(v'(v,w,n),w'(v,w,n))$ for each $n$}
\\
=\int_{(\mathbf R^3)^3}\big(B[\rho_f,\rho_g](x,y,v'(v,w,n_{x,y})-w'(v,w,n_{x,y}))-B[\rho_f,\rho_g](x,y,v-w)\big)
\\
\times \mathbf 1_{(v-w|n_{x,y})<0}f(x,v)g(y,w)dvdwdy=0&\,.
\\
\text{(since $(v'(v,w,n)-w'(v,w,n)|n)=-(v-w|n)$)}
\end{aligned}
\]
by assumption (ii) on $B$. Similarly
\[
\begin{aligned}
\int_{(\mathbf R^3)^2}\mathrm{St}[f,f](x,v)vdvdx
\\
=\int_{(\mathbf R^3)^4}\big(f(x,v'(v,w,n_{x,y}))f(y,w'(v,w,n_{x,y}))\mathbf 1_{(v-w|n_{x,y})>0}
\\
-f(x,v)f(y,w)\mathbf 1_{(v-w|n_{x,y})<0}\big)B[\rho_f,\rho_f](x,y,v-w)vdxdvdydw
\\
=\int_{(\mathbf R^3)^4}\big(B[\rho_f,\rho_f](x,y,v'(v,w,n_{x,y})-w'(v,w,n_{x,y}))v'(v,w,n_{x,y})
\\
-B[\rho_f,\rho_f](x,y,v-w)v\big)\mathbf 1_{(v-w|n_{x,y})<0}f(x,v)f(y,w)dxdvdydw
\\
\text{(with the change of variables $(v,w)\mapsto(v'(v,w,n),w'(v,w,n))$ for each $n$}
\\
\text{since $(v'(v,w,n)-w'(v,w,n)|n)=-(v-w|n)$)}
\\
=\int_{(\mathbf R^3)^4}\big(B[\rho_f,\rho_f](y,x,v'(w,v,n_{y,x})-w'(w,v,n_{y,x}))v'(w,v,n_{y,x})
\\
-B[\rho_f,\rho_f](y,x,w,v)w\big)\mathbf 1_{(w-v|n_{y,x})<0}f(x,v)f(y,w)dxdvdydw
\\
\text{(exchanging the molecule at $x$ with velocity $v$}
\\
\text{and the molecule at $y$ with velocity $w$ by assumption (i) on $B$)}
\\
=\int_{(\mathbf R^3)^4}\big(B[\rho_f,\rho_f](y,x,w'(v,w,n_{x,y})-v'(v,w,n_{x,y}))w'(v,w,n_{x,y})
\\
-B[\rho_f,\rho_f](y,x,w,v)w\big)\mathbf 1_{(w-v|n_{y,x})<0}f(x,v)f(y,w)dxdvdydw
\\
\text{(since $v'(w,v,n)=w'(v,w,n)=w'(v,w,-n)$}
\\
\text{and $w'(w,v,n)=v'(v,w,n)=v'(v,w,-n)$)}
\\
=\int_{(\mathbf R^3)^4}\big(B[\rho_f,\rho_f](x,y,v'(v,w,n_{x,y})-w'(v,w,n_{x,y}))w'(v,w,n_{x,y})
\\
-B[\rho_f,\rho_f](x,y,v,w)w\big)\mathbf 1_{(v-w|n_{x,y})<0}f(x,v)f(y,w)dxdvdydw&\,.
\\
\text{(using assumption (i) on $B$))}
\end{aligned}
\]
Summarizing, the first left-hand side is equal to the right-hand side of either the second or the sixth equality, and therefore to the mean of these expressions, i.e.
\[
\begin{aligned}
\int_{\mathbf R^3\times\mathbf R^3}\mathrm{St}[f,f](x,v)vdvdx
\\
=\int_{(\mathbf R^3)^4}\big(B[\rho_f,\rho_f](x,y,v'(v,w,n_{x,y})-w'(v,w,n_{x,y}))\tfrac{v'(v,w,n_{x,y})+w'(v,w,n_{x,y})}2
\\
-B[\rho_f,\rho_f](x,y,v-w)\tfrac{v+w}2\big)\mathbf 1_{(v-w|n_{x,y})<0}f(x,v)f(y,w)dxdvdydw=0
\end{aligned}
\]
by assumption (ii) on $B$, and because of the first identity in \eqref{MicroCons}.

Proceeding in exactly the same manner, we see that
\[
\begin{aligned}
\int_{\mathbf R^3\times\mathbf R^3}\mathrm{St}[f,f](x,v)|v|^2dvdx
\\
=\int_{(\mathbf R^3)^4}\big(B[\rho_f,\rho_f](x,y,v'(v,w,n_{x,y})-w'(v,w,n_{x,y}))\tfrac{|v'(v,w,n_{x,y})|^2+|w'(v,w,n_{x,y})|^2}2
\\
-B[\rho_f,\rho_f](x,y,v-w)\tfrac{|v|^2+|w|^2}2\big)\mathbf 1_{(v-w|n_{x,y})<0}f(x,v)f(y,w)dxdvdydw=0
\end{aligned}
\]
again by  assumption (ii) on $B$, and because of the second identity in \eqref{MicroCons}. This concludes the proof. 
\end{proof}

\subsection{Landau Currents for Delocalized Collision Integrals}\label{ssec:LandauCurrents}


By analogy with the arguments in \S41 in \cite{LL10} and in \cite{VillaniLandauCurr} leading to a representation of the Boltzmann collision integral as the divergence of a ``Landau'' mass current in the velocity
variable, our first main result in this paper is to represent delocalized collision integrals \eqref{GenCollInt} in terms of divergences of Landau mass, momentum and energy currents. The main difference with the 
Boltzmann case is that only the global conservation laws of momentum and energy are known to be satisfied by delocalized collision integrals \eqref{GenCollInt}. As a result, one should seek representations in 
terms of divergences of Landau momentum and energy currents \textit{in both the position and velocity variables}.

\begin{definition}\label{D-DisintMassCurrent}
For each rapidly decaying $f,g\in C(\mathbf R^3\times\mathbf R^3)$, and each rapidly decaying $\rho,\varpi\in C(\mathbf R^3)$, set
\[
\begin{aligned}
\mathfrak J_0[\rho,\varpi,f,g](x,y,v):=\int_{\mathbf R^3\times\mathbf R}B[\rho,\varpi](x,y,v-w)\mathbf 1_{0<s<(v-w|n_{y,x})}
\\
\times f(x,v+sn_{y,x})g(y,w+sn_{y,x})n_{y,x}dwds&\,.
\end{aligned}
\]
\end{definition}

Our first main result in this paper is the following theorem, whose proof can be found in section \ref{sec:ProofTCurrents}.

\begin{theorem}\label{T-Currents}
For each rapidly decaying $f\in C(\mathbf R^3\times\mathbf R^3)$,
\begin{equation}\label{DivMassCurrent}
\mathrm{St}[f,f]=\nabla_v\cdot\mathcal J_0[f]\,,\quad\text{ where }\mathcal J_0[f](x,v):=\int_{\mathbf R^3}\mathfrak J_0[\rho_f,\rho_f,f,f](x,y,v)dy\,.
\end{equation}
For $i=1,2,3$,
\begin{equation}\label{DivMomCurrents}
\mathrm{St}[f,f]v_i=\nabla_x\cdot\mathcal I_i[f]+\nabla_v\cdot\mathcal J_i[f]\,,
\end{equation}
where
\[
\begin{aligned}
\mathcal I_i[f](x,v):=-\tfrac12\int_{(\mathbf R^3)^2}\int_0^{\frac{\pi}2}\!B[\rho_f,\rho_f](x_{-\theta},y_{-\theta},v_{-\theta}\!-\!w_{-\theta})(v_{-\theta}\!-\!w_{-\theta}|n_{y_{-\theta},x_{-\theta}})_+
\\
\times f(x_{-\theta},v_{-\theta})f(y_{-\theta},w_{-\theta})(n_{y_{-\theta},x_{-\theta}})_iyd\theta dydw&
\end{aligned}
\]
where
\[
\left(\begin{matrix}x_\theta\\ y_\theta\end{matrix}\right):=\left(\begin{matrix}\cos\theta\,\,\,&\sin\theta\\ -\sin\theta&\cos\theta\end{matrix}\right)\left(\begin{matrix}x\\ y\end{matrix}\right)
\qquad\text{ and }\qquad
\left(\begin{matrix}v_\theta\\ w_\theta\end{matrix}\right):=\left(\begin{matrix}\cos\theta\,\,\,&\sin\theta\\ -\sin\theta&\cos\theta\end{matrix}\right)\left(\begin{matrix}v\\ w\end{matrix}\right)\,,
\]
and
\[
\begin{aligned}
\mathcal J_i[f,f](x,v):=&\int_{\mathbf R^3}\mathfrak J_0[\rho_f,\rho_f,(v_i-(v|n_{x,y})(n_{x,y})_i)f,f](x,y,v)dy
\\
&+\int_{\mathbf R^3}\mathfrak J_0[\rho_f,\rho_f,f,(w|n_{x,y})(n_{x,y})_if])(x,y,v)dy
\\
&-\!\tfrac12\int_{(\mathbf R^3)^4}\!\int_0^{\frac{\pi}2}\!\!B[\rho_f,\rho_f](x_{-\theta},y_{-\theta},v_{-\theta}\!-\!w_{-\theta})(v_{-\theta}\!-\!w_{-\theta}|n_{y_{-\theta},x_{-\theta}})_+
\\
&\qquad\qquad\qquad\times f(x_{-\theta},v_{-\theta})f(y_{-\theta},w_{-\theta})(n_{y_{-\theta},x_{-\theta}})_iwd\theta dydw\,.
\end{aligned}
\]
Finally
\begin{equation}\label{DivEnergyCurrent}
\mathrm{St}[f,f]|v|^2=\nabla_x\cdot\mathcal I_4[f]+\nabla_v\cdot\mathcal J_4[f]\,,
\end{equation}
where
\[
\begin{aligned}
\mathcal I_4[f](x,v)\!:=\!-\tfrac12\!\int_{(\mathbf R^3)^2}\!\int_0^{\frac{\pi}2}\!B[\rho_f,\rho_f](x_{-\theta},y_{-\theta},v_{-\theta}\!-\!w_{-\theta})(v_{-\theta}\!-\!w_{-\theta}|n_{y_{-\theta},x_{-\theta}})_+
\\
\times(v_{-\theta}\!+\!w_{-\theta}|n_{y_{-\theta},x_{-\theta}})f(x_{-\theta},v_{-\theta})f(y_{-\theta},w_{-\theta})yd\theta dydw&\,,
\end{aligned}
\]
and
\[
\begin{aligned}
\mathcal J_4[f](x,v):=&\int_{\mathbf R^3}\mathfrak J_0[\rho_f,\rho_f,(|v|^2\!-\!(v|n_{x,y})^2)f,f](x,y,v)dy
\\
&+\int_{\mathbf R^3}\mathfrak J_0[\rho_f,\rho_f,f,(w|n_{x,y})^2f](x,y,v)dy
\\
&-\tfrac12\int_{(\mathbf R^3)^2}\int_0^{\frac{\pi}2}\!B[\rho_f,\rho_f](x_{-\theta},y_{-\theta},v_{-\theta}\!-\!w_{-\theta})(v_{-\theta}\!-\!w_{-\theta}|n_{y_{-\theta},x_{-\theta}})_+
\\
&\quad\qquad\qquad\times\!\!(v_{-\theta}\!+\!w_{-\theta}|n_{y_{-\theta},x_{-\theta}})f(x_{-\theta},v_{-\theta})f(y_{-\theta},w_{-\theta})wd\theta dydw\,.
\end{aligned}
\]
\end{theorem}

\smallskip
There is some arbitrariness in the choice of the currents $\mathcal J_0$ and $\mathcal I_i,\mathcal J_i$ for $i=1,\ldots,4$. All these currents are obviously not unique, but defined up to arbitrary 
divergence-free vector fields. One can certainly try to eliminate this arbitrariness by means of some appropriate ``gauge'' condition, but the choice of such conditions is quite arbitrary, and we have 
chosen not to investigate this issue any further.

\subsection{H Theorem}\label{ssec:Entropy}


In general, delocalized collision integrals are not known to satisfy a property analogous to Boltzmann's H Theorem, which is well known in the case of the Boltzmann collision integral \eqref{BoltzCollInt}.
However, variants of the H Theorem are known in some cases: see \cite{Resibois,CerciArke,Polewczak,BelloLacho2} for the Enskog case, or \cite{CerciCPAM} for the soft sphere case. 

In this section, we shall assume that the collision kernel is of the form
\begin{equation}\label{CollKernelHThm}
B[\rho,\varpi](x,y,v-w)=b(|x-y|)|(v-w|n_{x,y})|
\end{equation}
where 
\[
b\in C([0,+\infty))\text{ satisfies }g\ge 0\text{ on }[0,+\infty)\text{ and }\mathrm{supp}\,b\subset[0,R]\,.
\]
We begin with a global Boltzmann type H Theorem --- except that the analogue of Boltzmann's H function is modified to take into account delocalization effects in the collision integral.

\begin{theorem}\label{T-HThmGlob}
Assume that $f(t,x,v)$ is a classical solution of 
\[
(\partial_t+v\cdot\nabla_x)f(t,x,v)=\mathrm{St}[f,f](t,x,v)\,,
\]
with $(x,v)\mapsto f(t,x,v)$ rapidly decaying at infinity, while $(x,v)\mapsto\ln f(t,x,v)$ has at most polynomial growth at infinity. Then
\[
\frac{d}{dt}(\mathbf H_B[f]+\mathbf H_\beta[\rho_f])=-\Lambda[f]\le 0\,,
\]
where $\mathbf H_B[f]$ is the Boltzmann $H$ functional
\[
\mathbf H_B[f](t):=\int_{(\mathbf R^3)^2}f(t,x,v)\ln f(t,x,v)dxdv\,,
\]
and $\mathbf H_\beta[\rho]$ the potential interaction energy
\[
\mathbf H_\beta[\rho](t):=\tfrac12\int_{(\mathbf R^3)^2}\rho_f(t,x)\rho_f(t,y)\beta(|x-y|)dxdy\,,
\]
while
\[
\begin{aligned}
\Lambda[f](t):=\tfrac12\int_{(\mathbf R^3)^4}f(t,x,v)f(t,y,w)\ell\left(\frac{f(t,x,v'(v,w,n_{y,x}))f(t,y,w'(v,w,n_{y,x}))}{f(t,x,v)f(t,y,w)}\right)
\\
\times b(|x-y|)(v-w|n_{y,x})_+dxdvdydw
\end{aligned}
\]
is the entropy production term. In these formulas, we have used the notations
\[
\beta(r):=\int_r^\infty b(s)ds\,,\qquad\ell(z)=z-1-\ln z\ge 0\text{ for all }z>0\,.
\]
\end{theorem}

This result is not really original: it is a modest generalization of the inequality (2.9) in \cite{CerciArke}. However, this generalization, especially the potential part $\mathbf H_\beta$, may have some important 
implications, to be discussed later.

\smallskip
The idea of conservative formulation of delocalized collision integral can be used in the context of the H Theorem. The idea is that the 

\begin{theorem}\label{T-HThmLoc}
Assume that $f(t,x,v)$ is a classical solution of 
\[
(\partial_t+v\cdot\nabla_x)f(t,x,v)=\mathrm{St}[f,f](t,x,v)\,,
\]
with $(x,v)\mapsto f(t,x,v)$ rapidly decaying at infinity, while $(x,v)\mapsto\ln f(t,x,v)$ has at most polynomial growth at infinity. Then
\begin{equation}\label{LocHIneq}
 \begin{aligned}
\mathrm{St}[f,f](t,x,v)\ln f(t,x,v)\le&-\nabla_x\cdot\mathcal K[f](t,x,v)-\nabla_v\cdot\mathcal L[f](t,x,v)
\\
&-\tfrac12(\partial_t+v\cdot\nabla_x)\left(f(t,x,v)\rho_f(t,\cdot)\star\beta(|\cdot|)(x)\right)\,,
\end{aligned}
\end{equation}
in the sense of distributions. In this inequality, the notation $\star$ designates the convolution (in the variable $x$), viz.
\begin{equation}\label{DefConvol}
\rho_f(t,\cdot)\star\beta(|\cdot|)(x):=\int_{\mathbf R^3}\rho_f(t,x-y)\beta(|y|)dy\,,
\end{equation}
while
\[
\begin{aligned}
\mathcal K[f](t,x,v)\!:=\!\tfrac12\int_{(\mathbf R^3)^2}\!\int_0^{\pi/2}\!(f(t,x_{-\theta},v'(v,w,n_{y,x})_{-\theta})f(t,y_{-\theta},w'(v,w,n_{y,x})_{-\theta})
\\
-f(t,x_{-\theta},v_{-\theta})f(t,y_{-\theta},w_{-\theta}))\ln f(t,y_{-\theta},w_{-\theta})b(|x_{-\theta}-y_{-\theta}|)
\\
\times (v_{-\theta}-w_{-\theta}|n_{x_{-\theta},y_{-\theta}})_+yd\theta dydw&\,,
\end{aligned}
\]
while
\[
\begin{aligned}
\mathcal L[f](t,x,v):=\tfrac12\mathcal J_0[f](t,x,v)(1-\rho_f\star\beta(|\cdot|)(t,x))
\\
+\tfrac12(\mathcal J_0[f\ln f,f](t,x,v)+\mathcal J_0[f,f\ln f](t,x,v))
\\
+\tfrac12\int_{(\mathbf R^3)^2}\int_0^{\pi/2}(f(t,x_{-\theta},v'(v,w,n_{y,x})_{-\theta})f(t,y_{-\theta},w'(v,w,n_{y,x})_{-\theta})
\\
-f(t,x_{-\theta},v_{-\theta})f(t,y_{-\theta},w_{-\theta}))\ln f(t,y_{-\theta},w_{-\theta})b(|x_{-\theta}-y_{-\theta}|)
\\
\times (v_{-\theta}-w_{-\theta}|n_{x_{-\theta},y_{-\theta}})_+wd\theta dydw&\,.
\end{aligned}
\]
In this equality, the notation $\mathcal J_0[f,g]$ designates the current
\[
\begin{aligned}
\mathcal J_0[f,g](t,x,v):=\int_{(\mathbf R^3)^2\times\mathbf R}b(|x-y|)(v-w|n_{y,x})\mathbf 1_{0<s<(v-w|n_{y,x})}
\\
\times f(t,x,v+sn_{y,x})g(t,y,w+sn_{y,x})n_{y,x}dydwds&\,.
\end{aligned}
\]
\end{theorem}

In other words
\[
\mathcal J_0[f,g](t,x,v)=\int_{\mathbf R^3}\mathfrak J_0[\rho_f,\rho_g,f,g](t,x,y,v)dy
\]
with $\mathfrak J_0[\rho_f,\rho_g,f,g]$ as in Definition \ref{D-DisintMassCurrent}. However, here, the situation is much simpler, since
\[
B[\rho_f,\rho_g](x,y,v-w):=b(|x-y|)|(v-w|n_{x,y})|
\]
is independent of $\rho_f$ and $\rho_g$, so that $\mathfrak J_0[\rho_f,\rho_g,f,g](t,x,y,v)$ is also independent of $\rho_f$ and $\rho_g$.

The proofs of Theorems 3.4 and 3.5 can be found in section \ref{sec:EntropIneq}.


\section{Applications to the Macroscopic Description of Dense Gases}\label{sec:Appli}


\subsection{Local Conservation of Mass, Momentum and Energy}


The discussion in section \ref{sec:LocConsLaws} makes it very clear that the Boltzmann collision integral does not contribute to the stress tensor in \eqref{EulEqBoltz}, i.e. the pressure tensor $P$. This may
seem somewhat disturbing, since one would expect that collisions between gas molecules should account for the acceleration driving the fluid. 

This last point is of course not entirely true. Indeed, even in a gas of noninteracting point particles, i.e. if the distribution function satisfies
\[
(\partial_t+v\cdot\nabla_x)f(t,x,v)=0\,,
\]
the pressure tensor $P$ defined in \eqref{DefP} can be computed explicitly, and it is neither equal to zero (unless $f(t,x,v)=\rho(t,x)\delta(v-u(t,x))$), nor a constant in general. It can be computed explicitly for 
instance in the case where
\[
f(0,x,v)=\exp(-(|x|^2+|v|^2))\,.
\]
One finds indeed that
\[
f(t,x,v)=\exp\left(-\left(\tfrac{|x|^2}{1+t^2}+(1+t^2)|v-\tfrac{t}{1+t^2}x|^2\right)\right)\,,
\]
so that 
\[
P(t,x)=\tfrac{\pi^{3/2}}{2(1+t^2)^{5/2}}\exp\left(-\tfrac{|x|^2}{1+t^2}\right)I\,.
\]
This elementary example shows that the force deriving from the pressure may be non trivial even in the absence of collisions.

In the case of the Boltzmann equation, the contribution of collisions to the pressure tensor is indirect: in the highly collisional regime, collisions induce a relaxation of the distribution function $f(t,x,v)$ to a local 
Maxwellian distribution, a radial function of $v-u(t,x)$, so that the current $Q$ vanishes identically, and $\rho,u$ defined in \eqref{DefRhoU} and $p:=\mathrm{tr}P$ is a solution of the system of equations 
\eqref{ContEq}-\eqref{EulEq}-\eqref{EnergEq}, i.e. the Euler system of gas dynamics for a perfect gas with adiabatic index $\gamma=5/3$, as explained at the end of section \ref{sec:LocConsLaws}. This
follows from the approach of the hydrodynamic limit of the Boltzmann equation involving Hilbert's expansion, as in \cite{CaflischCPAM80} or chapter 11 of \cite{CIP} --- see also \cite{BardosFG}.

In the case of dense gases, the situation is completely different. The collision integral \textit{directly} contributes to the stress tensor and energy flux in the local balance of momentum and energy in the gas, 
and this contribution is easily computed in terms of the momentum and energy currents in Theorem \ref{T-Currents}. See for instance \cite{Maynar,Takata} for other approaches to this issue.

\begin{corollary}\label{C-LocCons}
Let $f$ be a classical solution of the kinetic equation \eqref{GenKinMod} with delocalized collision integral \eqref{GenCollInt}. Assume that the collision kernel $B$ in \eqref{GenCollInt} satisfies (i)-(iv),
and that $(x,v)\mapsto f(t,x,v)$ is continuous on $\mathbf R^3\times\mathbf R^3$ and rapidly decaying at infinity. Then
\[
\left\{\begin{aligned}
{}&\partial_t\rho_f(t,x)+\nabla_x\cdot\int_{\mathbf R^3}vf(t,x,v)dv=0\,,
\\
{}&\partial_t\int_{\mathbf R^3}v_if(t,x,v)dv\!+\!\nabla_x\!\cdot\!\int_{\mathbf R^3}\left(vv_if(t,x,v)\!-\!\mathcal I_i[f](t,x,v)\right)dv=0\,,\qquad i=1,2,3\,,
\\
{}&\partial_t\int_{\mathbf R^3}\tfrac12|v|^2f(t,x,v)dv\!+\!\nabla_x\!\cdot\!\int_{\mathbf R^3}\left(v\tfrac12|v|^2f(t,x,v)-\tfrac12\mathcal I_4[f](t,x,v)\right)dv=0\,,
\end{aligned}
\right.
\]
where $\mathcal I_i[f]$ and $\mathcal I_4[f]$ are respectively the projections of the momentum and energy phase-space currents involved in \eqref{DivMomCurrents} and \eqref{DivEnergyCurrent} on the 
space of positions.
\end{corollary}

With the definitions \eqref{Def-rho} for $\rho_f$ and
\[
u_f(t,x):=\frac{\mathbf 1_{\rho_f(t,x)>0}}{\rho_f(t,x)}\int_{\mathbf R^3}vf(t,x,v)dv\,,
\]
and
\[
P_f(t,x):=\int_{\mathbf R^3}(v-u_f(t,x))^{\otimes 2}f(t,x,v)dv\,,
\]
while
\[
Q_f(t,x):=\tfrac12\int_{\mathbf R^3}(v-u_f(t,x))|v-u_f(t,x)|^2f(t,x,v)dv\,,
\]
as in section \ref{sec:LocConsLaws}, the identities obtained in the corollary above are recast as
\begin{equation}\label{LocConsLawsDeloc}
\left\{
\begin{aligned}
{}&\partial_t\rho_f+\nabla_x\cdot(\rho_fu_f)=0\,,
\\
&\partial_t(\rho_fu_f)+\nabla_x\cdot(\rho_f u_f^{\otimes 2}+P_f-\mathbb I[f])=0\,,
\\
&\partial_t(\tfrac12(\rho_f|u_f|^2+\mathrm{tr}P_f)+\nabla_x\cdot(\tfrac12(\rho_f|u_f|^2+\mathrm{tr}P_f)u_f+P_f\cdot u_f+Q_f-\mathcal I_4[f])=0\,,
\end{aligned}
\right.
\end{equation}
where $\mathbb I[f]$ is the square matrix whose rows are $\mathcal I_i[f]\in\mathbf R^3$, for $i=1,2,3$, i.e.
\[
\mathbb I[f]:=\left[\;\begin{matrix}{}\mathcal I_1[f]\\ \mathcal I_2[f]\\ \mathcal I_3[f]\end{matrix}\;\right]\,.
\]
Obviously the system \eqref{LocConsLawsDeloc} requires some closure relations --- i.e. expressions of $\mathbb I[f]$, of $Q_f$ and of $\mathcal I_4[f]$ in terms of $\rho_f,u_f$ and $P_f$ --- before it can 
be considered as a bona fide system of equations for the dynamics of dense gases.

\subsection{Local Entropy Inequality}


It is a well-known fact that systems of conservation laws such as \eqref{LocConsLawsDeloc} with some appropriate closure assumptions may have infinitely many weak solutions corresponding to the same 
initial data: see for instance section 4.4 in chapter IV of \cite{Dafermos}. The example of loss of uniqueness given in section 4.4 in chapter IV of \cite{Dafermos} involves unphysical shock waves, which 
dissipate, instead of producing entropy. For that reason, systems of conservation laws often come with a Friedrichs-Lax entropy inequality used to qualify weak solutions of such systems: see section 4.5
and chapter V in \cite{Dafermos}. (Unfortunately this is in general not enough to guarantee uniqueness in space dimension $2$ and higher: see \cite{DeLellisSzek}.)

However, these considerations suggest seeking an entropy inequality for \eqref{LocConsLawsDeloc}. By analogy with \cite{BardosFG}, we may think of using the local variant of the H Theorem reported
in Theorem \eqref{T-HThmLoc} to do so. 

\begin{corollary}\label{C-LocEntrIneq}
Let $f$ be a classical solution of \eqref{GenKinMod} with delocalized collision integral \eqref{GenCollInt} and collision kernel of the form \eqref{CollKernelHThm}. Assume that $(x,v)\mapsto f(t,x,v)$ is rapidly decaying at infinity, while $(x,v)\mapsto\ln f(t,x,v)$ has at most polynomial growth at infinity. Then, the local conservation laws \eqref{LocConsLawsDeloc} are accompanied with the entropy inequality
\[
\begin{aligned}
\partial_t\left(\int_{\mathbf R^3}f\ln f(t,x,v)dv+\tfrac12(\rho_f(t,\cdot)\star\beta(|\cdot|))(x)\rho_f(t,x)\right)
\\
+\nabla_x\cdot\left(\int_{\mathbf R^3}\left(f\ln f+\mathcal K[f]\right)(t,x,v)dv+\tfrac12(\rho_f(t,\cdot)\star\beta(|\cdot|))(x)\rho_f(t,x)u_f(t,x)\right)&\le 0
\end{aligned}
\]
where 
\[
\rho_f(t,x):=\int_{\mathbf R^3}f(t,x,v)dv\ge 0\,,\qquad u(t,x)=\frac{\mathbf 1_{\rho(t,x)>0}}{\rho(t,x)}\int_{\mathbf R^3}vf(t,x,v)dv\in\mathbf R^3\,,
\]
while the functional $\mathcal K$ is defined in Theorem \ref{T-HThmLoc}.
\end{corollary}

This is a straightforward consequence of Theorem \ref{T-HThmLoc}, obtained by integrating in $v\in\mathbf R^3$ both sides of the inequality in that theorem. To be rigorous, one should observe that the
inequality in Theorem \ref{T-HThmLoc}, which is stated in the sense of distributions, holds true when applied to a more general class of test functions, including $(x,v)\mapsto\psi(x)$ where $\psi$ belongs
to $C^1_c(\mathbf R^3)$: see the beginning of the proof of Theorem \ref{T-HThmLoc}.

The differential inequality obtained in Corollary \ref{C-LocEntrIneq} is not exactly an entropy inequality in the sense of Friedrichs and Lax, since the entropy
\[
\eta(t,x):=\int_{\mathbf R^3}f\ln f(t,x,v)dv+\tfrac12(\rho_f(t,\cdot)\star\beta(|\cdot|))(x)\rho_f(t,x)\,,
\]
even after some closure assumption specifying the $v$-dependence in $f$ (e.g. a local Maxwellian distribution), is not a \textit{local} function of the conserved densities in \eqref{LocConsLawsDeloc},
since it involves the term $\tfrac12(\rho_f(t,\cdot)\star\beta(|\cdot|))(x)\rho_f(t,x)$ which is nonlocal in $\rho_f$. 

Therefore, the usual theory of Friedrichs-Lax extensions \cite{FriedLax} of systems of conservation laws, and in particular the fact the convexity of $\eta$ as a \textit{local} function of the conserved density 
implies the hyperbolicity of the system of conservation laws, which can then be put in Godunov's form \cite{Godunov}, is unfortunately not available for us to handle the system \eqref{LocConsLawsDeloc}.
Nevertheless, even if one cannot apply directly the Godunov structure and the Friedrichs-Lax theory, the local entropy inequality reported in Corollary \ref{C-LocEntrIneq} may be of independent 
interest for studying \eqref{LocConsLawsDeloc}. For instance, it could be interesting to connect the discussion of the fluid dynamic limit of the Enskog-Boltzmann equation in \cite{Lachowicz} with
the results presented in this paper. We hope to return to this question in more detail in the future.


\section{Proof of Theorem \ref{T-Currents}}\label{sec:ProofTCurrents}


The proof of Theorem \ref{T-Currents} involves rather intricate computations, and will be split in several steps.

\subsection{Weak Forms of the Collision Integral}\label{ssec:WFormCollInt}

A first step in the representation of delocalized collision integrals in terms of divergences in $(x,v)$ of momentum and energy currents is the weak formulation of these collision integrals.

\begin{lemma}\label{L-WkFormStfg}
Let $f,g\in C(\mathbf R^3\times\mathbf R^3)$, rapidly decaying at infinity. For each test function $\phi\in C(\mathbf R^3\times\mathbf R^3)$ with (at most) polynomial growth at infinity,
\[
\begin{aligned}
\int_{(\mathbf R^3)^2}\!\mathrm{St}[f,g](x,v)\phi(x,v)dxdv\!=\!\!\int_{(\mathbf R^3)^4}\!(\phi(x,v'(v,w,n_{x,y}))\!-\!\phi(x,v))f(x,v)g(y,w)
\\
\times \mathbf 1_{(v-w|n_{x,y})<0}B[\rho_f,\rho_g](x,y,v-w)dxdvdydw&\,.
\end{aligned}
\]
\end{lemma}

\begin{proof}
By definition of the collision integral $\mathrm{St}[f,g]$,
\[
\begin{aligned}
\int_{(\mathbf R^3)^2}\mathrm{St}[f,g](x,v)\phi(x,v)dxdv
\\
=\int_{(\mathbf R^3)^4}B[\rho_f,\rho_g](x,y,v-w)(f(x,v'(v,w,n_{x,y}))g(y,w'(v,w,n_{x,y}))\mathbf 1_{(v-w|n_{x,y})>0}
\\
-f(x,v)g(y,w)\mathbf 1_{(v-w|n_{x,y})<0})\phi(x,v)dxdvdydw&\,.
\end{aligned}
\]
The gain term (positive part) of this collision integral satisfies the following symmetries:
\[
\begin{aligned}
\int_{(\mathbf R^3)^4}B[\rho_f,\rho_g](x,y,v-w)f(x,v'(v,w,n_{x,y}))g(y,w'(v,w,n_{x,y}))
\\
\times\phi(x,v)\mathbf 1_{(v-w|n_{x,y})>0}dxdvdydw
\\
=\int_{(\mathbf R^3)^4}B[\rho_f,\rho_g](x,y,v'(v,w,n_{x,y})-w'(v,w,n_{x,y}))f(x,v)g(y,w)
\\
\times\phi(x,v'(v,w,n_{x,y}))\mathbf 1_{(v'(v,w,n)-w'(v,w,n)|n_{x,y})>0}dxdvdydw
\\
\text{(integrating by substitution with $(v,w)\mapsto(v'(v,w,n_{x,y}),w'(v,w,n_{x,y}))$)}
\\
=\!\!\int_{(\mathbf R^3)^4}\!\!B[\rho_f,\rho_g](x,y,v\!-\!w)f(x,v)g(y,v)\phi(x,v'(v,w,n_{x,y}))\mathbf 1_{(v\!-\!w|n_{x,y})<0}dxdvdydw
\\
\text{(by (ii) and the equality $(v'(v,w,n_{x,y})-w'(v,w,n_{x,y})|n_{x,y})=-(v-w|n_{x,y})$).}
\end{aligned}
\]
Inserting this form of the gain term in the identity above implies the announced identity.
\end{proof}

\smallskip
The preceding lemma gives a weak formulation of $\mathrm{St}[f,g]$, where $f$ and $g$ may be different. In particular, the proof of this lemma does not involve property (i)
of the collision kernel $B$. The next lemma gives a weak formulation of $\mathrm{St}[f,f]$. Since $f=g$ in this case, property (i) of $B$ is a key ingredient in the proof.

\begin{lemma}\label{L-WeakFormStff}
Let $f\in C(\mathbf R^3\times\mathbf R^3)$ be rapidly decaying at infinity. For each test function $\phi\in C(\mathbf R^3\times\mathbf R^3)$ with polynomial growth at infinity,
\[
\begin{aligned}
\int_{(\mathbf R^3)^2}\mathrm{St}[f,f](x,v)\phi(x,v)dxdv
\\
=\tfrac12\int_{(\mathbf R^3)^4}(\phi(x,v'(v,w,n_{x,y}))+\phi(y,w'(v,w,n_{x,y}))-\phi(x,v)-\phi(y,w))
\\
\times f(x,v)f(y,w)B[\rho_f,\rho_f](x,y,v-w)\mathbf 1_{(v-w|n_{x,y})<0}dxdvdydw&\,.
\end{aligned}
\]
\end{lemma}

\begin{proof}
Starting from the identity in the preceding lemma
\begin{equation}\label{Ident1}
\begin{aligned}
\int_{(\mathbf R^3)^2}\!\mathrm{St}[f,f](x,v)\phi(x,v)dxdv\!=\!\int_{(\mathbf R^3)^4}(\phi(x,v'(v,w,n_{x,y}))\!-\!\phi(x,v))f(x,v)f(y,w)
\\
\times B[\rho_f,\rho_f](x,y,v\!-\!w)\mathbf 1_{(v-w|n_{x,y})<0}dxdvdydw&\,.
\end{aligned}
\end{equation}
Integrating by substitution in the right-hand side with $(x,v,y,w)\mapsto(y,w,x,v)$, 
\begin{equation}\label{Ident2}
\begin{aligned}
\int_{(\mathbf R^3)^2}\mathrm{St}[f,f](x,v)\phi(x,v)dxdv
\\
=\int_{(\mathbf R^3)^4}B[\rho_f,\rho_f](y,x,w-v)f(y,w)f(x,v)(\phi(y,v'(w,v,n_{y,x}))-\phi(y,w))
\\
\times\mathbf 1_{(w-v|n_{y,x})<0}dxdvdydwdn
\\
=\int_{(\mathbf R^3)^4}B[\rho_f,\rho_f](x,y,v-w)f(x,v)f(y,w)(\phi(y,w'(v,w,n_{x,y}))-\phi(y,w))
\\
\times \mathbf 1_{(v-w|n_{x,y})<0}dxdvdydw&\,,
\end{aligned}
\end{equation}
where the last equality follows from (i), while
\[
v'(w,v,-n)=v'(w,v,n)=w'(v,w,n)\,.
\]
Since the last right-hand sides of \eqref{Ident1} and \eqref{Ident2} are equal, they are both equal to their average, i.e.
\[
\begin{aligned}
\int_{(\mathbf R^3)^2}\mathrm{St}[f,f]\phi(x,v)dxdv=\int_{(\mathbf R^3)^4}\tfrac{\phi(x,v'(v,w,n_{x,y}))-\phi(x,v)+\phi(y,w'(v,w,n_{x,y}))-\phi(y,w)}2
\\
\times f(x,v)f(y,w)B[\rho_f,\rho_f](x,y,v-w)\mathbf 1_{(v-w|n_{x,y})<0}dxdvdydw&\,,
\end{aligned}
\]
which is precisely the announced identity.
\end{proof}

Lemma \ref{L-WkFormStfg} will be used repeatedly in the proof of Theorem \ref{T-Currents}. Lemma \ref{L-WeakFormStff}, on the other hand, is used in the proofs of Theorems  \ref{T-HThmGlob} and \ref{T-HThmLoc}.

\subsection{Mass Current: Proof of \eqref{DivMassCurrent}}

Observe that
\[
\begin{aligned}
\phi(x,v'(v,w,n))-\phi(x,v)=&\int_0^{(v-w|n)}\frac{d}{ds}\phi(x,v-sn)ds
\\
=&-\int_0^{(v-w|n)}n\cdot\nabla_v\phi(x,v-sn)ds\,.
\end{aligned}
\]

\[
\begin{aligned}
\int_{(\mathbf R^3)^2}\!\mathrm{St}[f,g](x,v)\phi(x,v)dxdv\!=\!\!\int_{(\mathbf R^3)^4}\!(\phi(x,v'(v,w,n_{x,y}))\!-\!\phi(x,v))f(x,v)g(y,w)
\\
\times \mathbf 1_{(v-w|n_{x,y})<0}B[\rho_f,\rho_g](x,y,v-w)dxdvdydw
\end{aligned}
\]

Inserting this expression in the right-hand side of the weak formulation of the collision integral $\mathrm{St}[f,g]$ in Lemma \ref{L-WkFormStfg} shows that
\[
\begin{aligned}
\int_{(\mathbf R^3)^2}\mathrm{St}[f,g](x,v)\phi(x,v)dxdv\!=\!-\int_{(\mathbf R^3)^4}\!B[\rho_f,\rho_g](x,y,v-w)f(x,v)g(y,w)
\\
\times\left(\int_0^{(v-w|n_{y,x})}n_{y,x}\cdot\nabla_v\phi(x,v-sn_{y,x})ds\right)\mathbf 1_{(v-w|n_{y,x})>0}dxdvdydw
\\
=-\int_{(\mathbf R^3)^4\times\mathbf R}B[\rho_f,\rho_g](x,y,v-w)\mathbf 1_{0<s<(v-w|n_{y,x})}f(x,v)g(y,w)
\\
\times n_{y,x}\cdot\nabla_v\phi(x,v-sn_{y,x})dxdvdydwds&\,.
\end{aligned}
\]
The last right-hand side above is recast as
\[
\begin{aligned}
\int_{(\mathbf R^3)^4\times\mathbf R}B[\rho_f,\rho_g](x,y,v-w)\mathbf 1_{0<s<(v-w|n_{y,x})}f(x,v)g(y,w)
\\
\times n_{y,x}\cdot\nabla_v\phi(x,v-sn_{y,x})dxdvdydwds
\\
=\int_{(\mathbf R^3)^4\times\mathbf R}B[\rho_f,\rho_g](x,y,\bar v-\bar w)f(x,\bar v+sn_{y,x})g(y,\bar w+sn_{y,x})
\\
\times\mathbf 1_{0<s<(\bar v-\bar w|n_{y,x})}n_{y,x}\cdot\nabla_v\phi(x,\bar v)dxd\bar vdyd\bar wds&\,,
\end{aligned}
\]
integrating by substitution with $\bar v:=v-sn_{y,x}$ and $\bar w:=w-sn_{y,x}$. Therefore
\[
\begin{aligned}
\int_{(\mathbf R^3)^2}\mathrm{St}[f,g](x,v)\phi(x,v)dxdv\!=\!-\!\int_{(\mathbf R^3)^3}\mathfrak J_0[\rho_f,\rho_g,f,g](x,y,\bar v)\cdot\nabla_v\phi(x,\bar v)dxdyd\bar v&\,,
\end{aligned}
\]
which is the formulation in the sense of distributions of the identity \eqref{DivMassCurrent}. \hfill{$\Box$}

\subsection{Momentum Currents: Proof of \eqref{DivMomCurrents}}

Applying Lemma \ref{L-WkFormStfg} to the test function $\phi(x,v)v_i$, we find that
\[
\begin{aligned}
\int_{(\mathbf R^3)^2}\!\mathrm{St}[f,g](x,v)\phi(x,v)v_idxdv\!=\!\!\int_{(\mathbf R^3)^4}\!(\phi(x,v'(v,w,n_{x,y}))v'_i(v,w,n_{x,y})\!-\!\phi(x,v)v_i)
\\
\times f(x,v)g(y,w)B[\rho_f,\rho_g](x,y,v-w)\mathbf 1_{(v-w|n_{x,y})<0}dxdvdydw&\,.
\end{aligned}
\]
Then we split the difference
\[
\begin{aligned}
\phi(x,v'(v,w,n_{x,y}))&v_i'(v,w,n_{x,y})-\phi(x,v)v_i
\\
=&(\phi(x,v'(v,w,n_{x,y}))-\phi(x,v))v_i'(v,w,n_{x,y})
\\
&+(v_i'(v,w,n_{x,y})-v_i)\phi(x,v)
\\
=&(\phi(x,v'(v,w,n_{x,y}))-\phi(x,v))v_i
\\
&-(\phi(x,v'(v,w,n_{x,y}))-\phi(x,v))(v|n_{x,y})(n_{x,y})_i
\\
&+(\phi(x,v'(v,w,n_{x,y}))-\phi(x,v))(w|n_{x,y})(n_{x,y})_i
\\
&-\phi(x,v)(v-w|n_{x,y})(n_{x,y})_i
\\
=&S_1-S_2+S_3-S_4\,.
\end{aligned}
\]
Thus
\[
\begin{aligned}
\int_{(\mathbf R^3)^4}B[\rho_f,\rho_f](x,y,v-w) \mathbf 1_{(v-w|n_{x,y})<0}f(x,v)f(y,w)S_1dxdvdydw
\\
=-\int_{(\mathbf R^3)^3}\mathfrak J_0[\rho_f,\rho_f,v_if,f](x,y,v)\cdot\nabla_v\phi(x,v)dxdydv&\,,
\end{aligned}
\]
while
\[
\begin{aligned}
\int_{(\mathbf R^3)^4}B[\rho_f,\rho_f](x,y,v-w) \mathbf 1_{(v-w|n_{x,y})<0}f(x,v)f(y,w)S_2dxdvdydw
\\
=-\int_{(\mathbf R^3)^3}\mathfrak J_0[\rho_f,\rho_f,(v|n_{x,y})(n_{x,y})_if,f](x,y,v)\cdot\nabla_v\phi(x,v)dxdydv&\,,
\end{aligned}
\]
and
\[
\begin{aligned}
\int_{(\mathbf R^3)^4}B[\rho_f,\rho_f](x,y,v-w) \mathbf 1_{(v-w|n_{x,y})<0}f(x,v)f(y,w)S_3dxdvdydw
\\
=-\int_{(\mathbf R^3)^3}\mathfrak J_0[\rho_f,\rho_f,f,(w|n_{x,y})(n_{x,y})_if](x,y,v)\cdot\nabla_v\phi(x,v)dxdydv&\,.
\end{aligned}
\]
We are left with the term
\[
\begin{aligned}
\int_{(\mathbf R^3)^4}B[\rho_f,\rho_f](x,y,v-w) \mathbf 1_{(v-w|n_{x,y})<0}f(x,v)f(y,w)S_4dxdvdydw
\\
=\int_{(\mathbf R^3)^4}B[\rho_f,\rho_f](x,y,v\!-\!w)(v\!-\!w|n_{y,x})_+(n_{y,x})_i\phi(x,v)f(x,v)f(y,w)dxdvdydw
\\
=\int_{(\mathbf R^3)^4}B[\rho_f,\rho_f](y,x,w\!-\!v)(v\!-\!w|n_{y,x})_+(n_{y,x})_i\phi(x,v)f(x,v)f(y,w)dxdvdydw&\,,
\end{aligned}
\]
where the last equality follows from assumption (i) on the collision kernel $B$. This last integral is transformed with the change of variables 
\[
(x,y,v,w)\mapsto(y,x,w,v)
\]
so that
\[
\begin{aligned}
\int_{(\mathbf R^3)^4}B[\rho_f,\rho_f](y,x,w\!-\!v)(v\!-\!w|n_{y,x})_+(n_{y,x})_i\phi(x,v)f(x,v)f(y,w)dxdvdydw
\\
=\int_{(\mathbf R^3)^4}B[\rho_f,\rho_f](x,y,v\!-\!w)(w\!-\!v|n_{x,y})_+(n_{x,y})_i\phi(y,w)f(x,v)f(y,w)dxdvdydw
\\
=-\int_{(\mathbf R^3)^4}B[\rho_f,\rho_f](x,y,v\!-\!w)(v\!-\!w|n_{y,x})_+(n_{y,x})_i\phi(y,w)f(x,v)f(y,w)dxdvdydw&\,.
\end{aligned}
\]
Thus
\[
\begin{aligned}
\int_{(\mathbf R^3)^4}B[\rho_f,\rho_f](x,y,v-w) \mathbf 1_{(v-w|n_{x,y})<0}f(x,v)f(y,w)S_4dxdvdydw
\\
=\int_{(\mathbf R^3)^4}B[\rho_f,\rho_f](x,y,v\!-\!w)(v\!-\!w|n_{y,x})_+(n_{y,x})_i\phi(x,v)f(x,v)f(y,w)dxdvdydw
\\
=-\int_{(\mathbf R^3)^4}B[\rho_f,\rho_f](x,y,v\!-\!w)(v\!-\!w|n_{y,x})_+(n_{y,x})_i\phi(y,w)f(x,v)f(y,w)dxdvdydw&\,,
\end{aligned}
\]
so that, eventually,
\[
\begin{aligned}
\int_{(\mathbf R^3)^4}B[\rho_f,\rho_f](x,y,v-w) \mathbf 1_{(v-w|n_{x,y})<0}f(x,v)f(y,w)S_4dxdvdydw
\\
=\tfrac12\int_{(\mathbf R^3)^4}B[\rho_f,\rho_f](x,y,v\!-\!w)(v\!-\!w|n_{y,x})_+(n_{y,x})_i(\phi(x,v)-\phi(y,w))
\\
\times f(x,v)f(y,w)dxdvdydw&\,.
\end{aligned}
\]

Set
\[
\left(\begin{matrix}x_\theta\\ y_\theta\end{matrix}\right):=\left(\begin{matrix}\cos\theta\,\,\,&\sin\theta\\ -\sin\theta&\cos\theta\end{matrix}\right)\left(\begin{matrix}x\\ y\end{matrix}\right)
\qquad\text{ and }\qquad
\left(\begin{matrix}v_\theta\\ w_\theta\end{matrix}\right):=\left(\begin{matrix}\cos\theta\,\,\,&\sin\theta\\ -\sin\theta&\cos\theta\end{matrix}\right)\left(\begin{matrix}v\\ w\end{matrix}\right)\,.
\]
We shall express the difference $\phi(x,v)-\phi(y,w)$ as follows:
\[
\phi(x,v)-\phi(y,w)=-\int_0^{\pi/2}\frac{d}{d\theta}\phi(x_\theta,v_\theta)d\theta\,.
\]
Since
\[
\frac{d}{d\theta}\phi(x_\theta,v_\theta)=y_\theta\cdot\nabla_x\phi(x_\theta,v_\theta)+w_\theta\cdot\nabla_v\phi(x_\theta,v_\theta)\,,
\]
we find that
\[
\begin{aligned}
\int_{(\mathbf R^3)^4}B[\rho_f,\rho_f](x,y,v-w) \mathbf 1_{(v-w|n_{x,y})<0}f(x,v)f(y,w)S_4dxdvdydw
\\
=-\tfrac12\int_{(\mathbf R^3)^4}\int_0^{\pi/2}B[\rho_f,\rho_f](x,y,v\!-\!w)(v\!-\!w|n_{y,x})_+(n_{y,x})_iy_\theta\cdot\nabla_x\phi(x_\theta,v_\theta)
\\
\times f(x,v)f(y,w)d\theta dxdvdydw
\\
-\tfrac12\int_{(\mathbf R^3)^4}\int_0^{\pi/2}B[\rho_f,\rho_f](x,y,v\!-\!w)(v\!-\!w|n_{y,x})_+(n_{y,x})_iw_\theta\cdot\nabla_v\phi(x_\theta,v_\theta)
\\
\times f(x,v)f(y,w)d\theta dxdvdydw&\,.
\end{aligned}
\]
Since the Jacobian determinants
\[
\det\frac{\partial(x_\theta,y_\theta)}{\partial(x,y)}=\det\frac{\partial(v_\theta,w_\theta)}{\partial(v,w)}=1\,,
\]
one can transforms both integrals in the right-hand side of the identity above into
\[
\begin{aligned}
-\int_{(\mathbf R^3)^4}B[\rho_f,\rho_f](x,y,v-w) \mathbf 1_{(v-w|n_{x,y})<0}f(x,v)f(y,w)S_4dxdvdydw
\\
=\tfrac12\int_{(\mathbf R^3)^4}\nabla_x\phi(x,v)\cdot y\int_0^{\pi/2}B[\rho_f,\rho_f](x_{-\theta},y_{-\theta},v_{-\theta}\!-\!w_{-\theta})(v_{-\theta}\!-\!w_{-\theta}|n_{y_{-\theta},x_{-\theta}})_+
\\
\times(n_{y_{-\theta},x_{-\theta}})_if(x_{-\theta},v_{-\theta})f(y_{-\theta},w_{-\theta})d\theta dxdvdydw
\\
+\tfrac12\int_{(\mathbf R^3)^4}\nabla_v\phi(x,v)\cdot w\int_0^{\pi/2}B[\rho_f,\rho_f](x_{-\theta},y_{-\theta},v_{-\theta}\!-\!w_{-\theta})(v_{-\theta}\!-\!w_{-\theta}|n_{y_{-\theta},x_{-\theta}})_+
\\
\times(n_{y_{-\theta},x_{-\theta}})_if(x_{-\theta},v_{-\theta})f(y_{-\theta},w_{-\theta})d\theta dxdvdydw&\,.
\end{aligned}
\]
Summarizing, we have found that
\[
\begin{aligned}
\int_{(\mathbf R^3)^2}\mathrm{St}[f,f](x,v)\phi(x,v)v_kdxdv=-\int_{(\mathbf R^3)^2}\mathcal I_i[f](x,v)\cdot\nabla_x\phi(x,v)dxdv
\\
-\int_{(\mathbf R^3)^2}\mathcal J_i[f](x,v)\cdot\nabla_v\phi(x,v)dxdv&\,,
\end{aligned}
\]
where
\[
\begin{aligned}
\mathcal I_i[f](x,v):=-\tfrac12\int_{(\mathbf R^3)^4}\int_0^{\pi/2}B[\rho_f,\rho_f](x_{-\theta},y_{-\theta},v_{-\theta}\!-\!w_{-\theta})(v_{-\theta}\!-\!w_{-\theta}|n_{y_{-\theta},x_{-\theta}})_+
\\
\times f(x_{-\theta},v_{-\theta})f(y_{-\theta},w_{-\theta})(n_{y_{-\theta},x_{-\theta}})_iyd\theta dydw&\,,
\end{aligned}
\]
while
\[
\begin{aligned}
\mathcal J_i[f](x,v):=\int_{\mathbf R^3}\mathfrak J_0[\rho_f,\rho_f,(v_i-(v|n_{x,y})(n_{x,y})_i)f,f](x,y,v)dy
\\
+\int_{\mathbf R^3}\mathfrak J_0[\rho_f,\rho_f,f,(w|n_{x,y})(n_{x,y})_if])(x,y,v)dy
\\
-\tfrac12\int_{(\mathbf R^3)^4}\int_0^{\pi/2}B[\rho_f,\rho_f](x_{-\theta},y_{-\theta},v_{-\theta}\!-\!w_{-\theta})(v_{-\theta}\!-\!w_{-\theta}|n_{y_{-\theta},x_{-\theta}})_+
\\
\times f(x_{-\theta},v_{-\theta})f(y_{-\theta},w_{-\theta})(n_{y_{-\theta},x_{-\theta}})_iwd\theta dydw&\,.
\end{aligned}
\]
This is the weak formulation (in the sense of distributions) of \eqref{DivMomCurrents}. \hfill$\Box$

\subsection{Energy Currents: Proof of \eqref{DivEnergyCurrent}}

Applying Lemma \ref{L-WkFormStfg} to the test function $\phi(x,v)|v|^2$, we find that
\[
\begin{aligned}
\int_{(\mathbf R^3)^2}\!\mathrm{St}[f,g](x,v)\phi(x,v)|v|^2dxdv
\\
=\int_{(\mathbf R^3)^4}\!(\phi(x,v'(v,w,n_{x,y}))|v'(v,w,n_{x,y})|^2\!-\!\phi(x,v)|v|^2)f(x,v)g(y,w)
\\
\times B[\rho_f,\rho_g](x,y,v-w)\mathbf 1_{(v-w|n_{x,y})<0}dxdvdydw
\end{aligned}
\]
Observe that
\[
v'(v,w,n)=(v-(v|n)n)+(w|n)n\in(\mathbf Rn)^\perp+\mathbf Rn\,.
\]
Hence
\[
|v'(v,w,n)|^2=|v-(v|n)n|^2+(w|n)^2=|v|^2-(v|n)^2+(w|n)^2\,.
\]
Then we split the difference
\[
\begin{aligned}
\phi(x,v'(v,w,n_{x,y}))&|v'(v,w,n_{x,y})|^2-\phi(x,v)|v|^2
\\
=&(\phi(x,v'(v,w,n_{x,y}))-\phi(x,v))|v'(v,w,n_{x,y})|^2
\\
&+(|v'(v,w,n_{x,y})|^2-|v|^2)\phi(x,v)
\\
=&(\phi(x,v'(v,w,n_{x,y}))-\phi(x,v))|v|^2
\\
&-(\phi(x,v'(v,w,n_{x,y}))-\phi(x,v))|(v|n_{x,y})|^2
\\
&+(\phi(x,v'(v,w,n_{x,y}))-\phi(x,v))|(w|n_{x,y})|^2
\\
&-\phi(x,v)((v|n_{x,y})^2-(w|n_{x,y})^2)
\\
=&S_1-S_2+S_3-S_4\,.
\end{aligned}
\]
Proceeding as in the case of the momentum current, we find that
\[
\begin{aligned}
\int_{(\mathbf R^3)^4}B[\rho_f,\rho_f](x,y,v-w) \mathbf 1_{(v-w|n_{x,y})<0}f(x,v)f(y,w)
\\
\times(S_1-S_2+S_3)dxdvdydw
\\
=-\int_{(\mathbf R^3)^3}\mathfrak J_0[\rho_f,\rho_f,(|v|^2-(v|n_{x,y})^2)f,f](x,y,v)\cdot\nabla_v\phi(x,v)dxdydv
\\
-\int_{(\mathbf R^3)^3}\mathfrak J_0[\rho_f,\rho_f,f,(w|n_{x,y})^2f](x,y,v)\cdot\nabla_v\phi(x,v)dxdydv&\,.
\end{aligned}
\]
Then we are left with the term
\[
\begin{aligned}
-\int_{(\mathbf R^3)^4}B[\rho_f,\rho_f](x,y,v-w) \mathbf 1_{(v-w|n_{x,y})<0}f(x,v)f(y,w)S_4dxdvdydw
\\
=\int_{(\mathbf R^3)^4}B[\rho_f,\rho_f](x,y,v-w) \mathbf 1_{(v-w|n_{x,y})<0}((v|n_{x,y})^2-(w|n_{x,y})^2)\phi(x,v)
\\
\times f(x,v)f(y,w)dxdvdydw
\\
=\int_{(\mathbf R^3)^4}B[\rho_f,\rho_f](x,y,v-w)(v-w|n_{y,x})_+(v+w|n_{y,x})\phi(x,v)
\\
\times f(x,v)f(y,w)dxdvdydw
\\
=-\int_{(\mathbf R^3)^4}B[\rho_f,\rho_f](y,x,w-v)(w-v|n_{x,y})_+(w+v|n_{x,y})\phi(x,v)
\\
\times f(x,v)f(y,w)dxdvdydw&\,.
\end{aligned}
\]
by assumption (i) on the collision kernel $B$, observing that $n_{y,x}=-n_{x,y}$. With the change of variables 
\[
(x,v,y,w)\mapsto(y,w,x,v)
\]
this last integral is transformed into
\[
\begin{aligned}
-\int_{(\mathbf R^3)^4}B[\rho_f,\rho_f](x,y,v-w) \mathbf 1_{(v-w|n_{x,y})<0}f(x,v)f(y,w)S_4dxdvdydw
\\
=-\int_{(\mathbf R^3)^4}B[\rho_f,\rho_f](x,y,v-w)(v-w|n_{y,x})_+(v+w|n_{y,x})\phi(y,w)
\\
\times f(x,v)f(y,w)dxdvdydw&\,.
\end{aligned}
\]
Therefore
\[
\begin{aligned}
-\int_{(\mathbf R^3)^4}B[\rho_f,\rho_f](x,y,v-w) \mathbf 1_{(v-w|n_{x,y})<0}f(x,v)f(y,w)S_4dxdvdydw
\\
=\tfrac12\int_{(\mathbf R^3)^4}B[\rho_f,\rho_f](x,y,v-w)(v-w|n_{y,x})_+(v+w|n_{y,x})(\phi(x,v)-\phi(y,w))
\\
\times f(x,v)f(y,w)dxdvdydw&\,.
\end{aligned}
\]
With the same notation as in the preceding section, i.e.
\[
\left(\begin{matrix}x_\theta\\ y_\theta\end{matrix}\right):=\left(\begin{matrix}\cos\theta\,\,\,&\sin\theta\\ -\sin\theta&\cos\theta\end{matrix}\right)\left(\begin{matrix}x\\ y\end{matrix}\right)
\qquad\text{ and }\qquad
\left(\begin{matrix}v_\theta\\ w_\theta\end{matrix}\right):=\left(\begin{matrix}\cos\theta\,\,\,&\sin\theta\\ -\sin\theta&\cos\theta\end{matrix}\right)\left(\begin{matrix}v\\ w\end{matrix}\right)\,,
\]
the difference $\phi(x,v)-\phi(y,w)$ is represented as
\[
\phi(y,w)-\phi(x,v)=\int_0^{\pi/2}(y_\theta\cdot\nabla_x\phi+w_\theta\cdot\nabla_v\phi)(x_\theta,v_\theta)d\theta\,.
\]
Substituting this expression in the identity above leads to
\[
\begin{aligned}
-\int_{(\mathbf R^3)^4}B[\rho_f,\rho_f](x,y,v-w) \mathbf 1_{(v-w|n_{x,y})<0}f(x,v)f(y,w)S_4dxdvdydwdn
\\
=\tfrac12\int_{(\mathbf R^3)^4}\int_0^{\pi/2}B[\rho_f,\rho_f](x,y,v-w)(y_\theta\cdot\nabla_x\phi+w_\theta\cdot\nabla_v\phi)(x_\theta,v_\theta)
\\
\times(v-w|n_{y,x})_+(v+w|n_{y,x})f(x,v)f(y,w)d\theta dxdvdydw
\\
=\tfrac12\int_{(\mathbf R^3)^4}\int_0^{\pi/2}B[\rho_f,\rho_f](x_{-\theta},y_{-\theta},v_{-\theta}\!-\!w_{-\theta})(y\cdot\nabla_x\phi+w\cdot\nabla_v\phi)(x,v)
\\
\times(v_{-\theta}\!-\!w_{-\theta}|n_{y_{-\theta},x_{-\theta}})_+(v_{-\theta}\!+\!w_{-\theta}|n_{y_{-\theta},x_{-\theta}})f(x_{-\theta},v_{-\theta})f(y_{-\theta},w_{-\theta})d\theta dxdvdydw&\,.
\end{aligned}
\]

Summarizing, we have proved that
\[
\begin{aligned}
\int_{(\mathbf R^3)^2}\mathrm{St}[f,f](x,v)\phi(x,v)|v|^2dxdv=-\int_{(\mathbf R^3)^2}\mathcal I_4[f](x,v)\cdot\nabla_x\phi(x,v)dxdv
\\
-\int_{(\mathbf R^3)^2}\mathcal J_4[f](x,v)\cdot\nabla_v\phi(x,v)dxdv&\,,
\end{aligned}
\]
where
\[
\begin{aligned}
\mathcal I_4[f](x,v):=-\tfrac12\int_{(\mathbf R^3)^2}\int_0^{\pi/2}B[\rho_f,\rho_f](x_{-\theta},y_{-\theta},v_{-\theta}\!-\!w_{-\theta})
\\
\times(v_{-\theta}\!-\!w_{-\theta}|n_{y_{-\theta},x_{-\theta}})_+(v_{-\theta}\!+\!w_{-\theta}|n_{y_{-\theta},x_{-\theta}})f(x_{-\theta},v_{-\theta})f(y_{-\theta},w_{-\theta})yd\theta dydw&\,.
\end{aligned}
\]
while
\[
\begin{aligned}
\mathcal J_4[f](x,v):=&-\int_{\mathbf R^3}\mathfrak J_0[\rho_f,\rho_f,(|v|^2-(v|n_{x,y})^2)f,f](x,y,v)dy
\\
&-\int_{\mathbf R^3}\mathfrak J_0[\rho_f,\rho_f,f,(w|n_{x,y})^2f](x,y,v)dy
\\
&-\tfrac12\int_{(\mathbf R^3)^2}\int_0^{\pi/2}B[\rho_f,\rho_f](x_{-\theta},y_{-\theta},v_{-\theta}\!-\!w_{-\theta})(v_{-\theta}\!-\!w_{-\theta}|n_{y_{-\theta},x_{-\theta}})_+
\\
&\times(v_{-\theta}+w_{-\theta}|n_{y_{-\theta},x_{-\theta}})f(x_{-\theta},v_{-\theta})f(y_{-\theta},w_{-\theta})wd\theta dydw\,.
\end{aligned}
\]
This is the weak formulation (in the sense of distributions) of \eqref{DivEnergyCurrent}. \hfill$\Box$


\section{Entropy Inequalities for Delocalized Collision Integrals}\label{sec:EntropIneq}


\subsection{Proof of Theorem \ref{T-HThmGlob}}\label{ssec:ProofHThmGlob}


Start from the weak formulation of the collision integral in Lemma \ref{L-WeakFormStff}, with test function $\phi:=\ln f+1$:
\[
\begin{aligned}
\int_{(\mathbf R^3)^2}\mathrm{St}[f,f](t,x,v)(\ln f(t,x,v)+1)dxdv
\\
=\tfrac12\int_{(\mathbf R^3)^4}f(t,x,v)f(t,y,w)\ln\frac{f(t,x,v'(v,w,n_{y,x})f(t,y,w'(v,w,n_{y,x}))}{f(t,x,v)f(t,y,w)}
\\
\times b(|x-y|)(v-w|n_{y,x})_+dxdvdydw&\,.
\end{aligned}
\]
Thus, in terms of the function $\ell$,
\[
\begin{aligned}
\int_{(\mathbf R^3)^2}\mathrm{St}[f,f](t,x,v)(\ln f(t,x,v)+1)dxdv
\\
=-\tfrac12\int_{(\mathbf R^3)^4}f(t,x,v)f(t,y,w)\ell\left(\frac{f(t,x,v'(v,w,n_{y,x}))f(t,y,w'(v,w,n_{y,x}))}{f(t,x,v)f(t,y,w)}\right)
\\
\times b(|x-y|)(v-w|n_{y,x})_+dxdvdydw
\\
+\tfrac12\int_{(\mathbf R^3)^4}(f(t,x,v'(v,w,n_{y,x}))f(t,y,w'(v,w,n_{y,x}))-f(t,x,v)f(t,y,w))
\\
\times b(|x-y|)(v-w|n_{y,x})_+dxdvdydw
\\
=S_1+S_2&\,.
\end{aligned}
\]
Obviously $S_1\le 0$, while 
\[
\begin{aligned}
\int_{(\mathbf R^3)^4}\!\!f(t,x,v'(v,w,n_{y,x}))f(t,y,w'(v,w,n_{y,x}))b(|x-y|)(v-w|n_{y,x})_+dxdvdydw
\\
=\!\int_{(\mathbf R^3)^4}\!\!f(t,x,v)f(t,y,v)b(|x\!-\!y|)(v'(v,w,n_{y,x})\!-\!w'(v,w,n_{y,x})|n_{y,x})_+dxdvdydw
\\
=\int_{(\mathbf R^3)^4}\!\!f(t,x,v)f(t,y,v)b(|x-y|)(v-w|n_{x,y})_+dxdvdydw&\,.
\end{aligned}
\]
Hence
\[
\begin{aligned}
S_2\!=\!\tfrac12\int_{(\mathbf R^3)^4}\!f(t,x,v)f(t,y,w)((v\!-\!w|n_{x,y})_+\!-\!(v\!-\!w|n_{x,y})_-)b(|x-y|)dxdvdydw
\\
=\tfrac12\int_{(\mathbf R^3)^4}f(t,x,v)f(t,y,w)(v-w|n_{x,y})b(|x-y|)dxdvdydw&\,.
\end{aligned}
\]
The formula
\[
\beta(r):=\int_r^\infty b(s)ds
\]
implies that
\begin{equation}\label{PotDer}
\begin{aligned}
(v\cdot\nabla_x+w\cdot\nabla_y)\beta(|x-y|)&=(v-w|\tfrac{x-y}{|x-y|})\beta'(|x-y|)
\\
&=-(v-w|n_{x,y})b(|x-y|)\,.
\end{aligned}
\end{equation}
Hence
\[
\begin{aligned}
\int_{(\mathbf R^3)^4}f(t,x,v)f(t,y,w)(v-w|n_{x,y})b(|x-y|)dxdvdydw
\\
=\int_{(\mathbf R^3)^4}(v\cdot\nabla_x+w\cdot\nabla_y)(f(t,x,v)f(t,y,w))\beta(|x-y|)dxdvdydw
\\
=\int_{(\mathbf R^3)^4}(\mathrm{St}[f,f](t,x,v)f(t,y,w)+f(t,x,v)\mathrm{St}[f,f](t,y,w))\beta(|x-y|)dxdvdydw
\\
-\int_{(\mathbf R^3)^4}(\partial_tf(t,x,v)f(t,y,w)+f(t,x,v)\partial_tf(t,y,w))\beta(|x-y|)dxdvdydw&\,.
\end{aligned}
\]
Observing that the local mass conservation implies that
\[
\int_{\mathbf R^3}\mathrm{St}[f,f](t,x,v)dv=\int_{\mathbf R^3}\mathrm{St}[f,f](t,y,w)dw=0\,,
\]
we conclude that
\[
S_2=-\tfrac12\frac{d}{dt}\int_{(\mathbf R^3)^4}f(t,x,v)f(t,y,w)\beta(|x-y|)dxdvdydw\,.
\]
Summarizing, we have proved that
\[
\frac{d}{dt}\int_{(\mathbf R^3)^2}f\ln f(t,x,v)dxdv=\int_{(\mathbf R^3)^2}\mathrm{St}[f,f](t,x,v)(\ln f(t,x,v)+1)dxdv=S_1+S_2
\]
and the conclusion follows from the expression of $S_2$ obtained above, after defining $\Lambda$ by the formula $\Lambda:=-S_1\ge 0$. \hfill$\Box$

\subsection{Proof of Theorem \ref{T-HThmLoc}}\label{ssec:ProofHThmLoc}


This proof involves somewhat lengthy computations, and will be split in six steps for the reader's convenience. Throughout this proof, $\psi$ designates a test function satisfying the following
conditions
\[
\psi\in C^1(\mathbf R^3\times\mathbf R^3)\cap L^\infty(\mathbf R^3\times\mathbf R^3)\,,\qquad\psi(x,v)\ge 0\text{ for all }x,v\in\mathbf R^3\,.
\]

\noindent
\underline{Step 1.} Start from the weak formulation of the collision integral in Lemma \ref{L-WeakFormStff}, with test function $\phi:=\psi\ln f$: 
observing that $v'(v,w,n)=v'(v,w,-n)$ and $w'(v,w,n)=w'(v,w,-n)$,
\[
\begin{aligned}
\int_{(\mathbf R^3)^2}\mathrm{St}[f,f](t,x,v)\ln f(t,x,v)\psi(x,v)dxdv
\\
=\int_{(\mathbf R^3)^4}\tfrac{(\psi\ln f)(t,x,v'(v,w,n_{y,x}))+(\psi\ln f)(t,y,w'(v,w,n_{y,x}))-(\psi\ln f)(t,x,v)-(\psi\ln f)(t,y,w)}2
\\
\times f(t,x,v)f(t,y,w)b(|x-y|)(v-w|n_{y,x})_+dxdvdydw&\,.
\end{aligned}
\]
Decompose
\[
S:=\tfrac{(\psi\ln f)(t,x,v'(v,w,n_{y,x}))+(\psi\ln f)(t,y,w'(v,w,n_{y,x}))-(\psi\ln f)(t,x,v)-(\psi\ln f)(t,y,w)}2
\]
as
\[
\begin{aligned}
S=\tfrac12\psi(x,v'(v,w,n_{y,x}))\ln\frac{f(t,x,v'(v,w,n_{y,x}))f(t,y,w'(v,w,n_{y,x}))}{f(t,x,v)f(t,y,w)}
\\
+\tfrac12(\psi(y,w'(v,w,n_{y,x}))-\psi(x,v'(v,w,n_{y,x}))\ln f(t,y,w'(v,w,n_{y,x}))
\\
-\tfrac12(\psi(x,v))-\psi(x,v'(v,w,n_{y,x}))\ln f(t,x,v)
\\
-\tfrac12(\psi(y,w)-\psi(x,v'(v,w,n_{y,x}))\ln f(t,y,w)
\\
=S_1+S_2-S_3-S_4&\,.
\end{aligned}
\]
The term $S_4$ is decomposed further into
\[
\begin{aligned}
S_4=&\tfrac12(\psi(y,w)-\psi(x,v'(v,w,n_{y,x}))\ln f(t,y,w)
\\
=&\tfrac12(\psi(y,w)-\psi(x,v))\ln f(t,y,w)
\\
&+\tfrac12(\psi(x,v)-\psi(x,v'(v,w,n_{y,x}))\ln f(t,y,w)
\\
=&S_{41}+S_{42}\,.
\end{aligned}
\]
\underline{Step 2.} Similarly, with $\ell(z):=z-1-\ln z\ge 0$ for $z>0$,
\[
\begin{aligned}
S_1=&-\tfrac12\psi(x,v'(v,w,n_{y,x}))\ell\left(\frac{f(t,x,v'(v,w,n_{y,x}))f(t,y,w'(v,w,n_{y,x}))}{f(t,x,v)f(t,y,w)}\right)
\\
&+\tfrac12\psi(x,v'(v,w,n_{y,x}))\left(\frac{f(t,x,v'(v,w,n_{y,x}))f(t,y,w'(v,w,n_{y,x}))}{f(t,x,v)f(t,y,w)}-1\right)
\\
=&-S_{11}+S_{12}\,.
\end{aligned}
\]
Thus
\[
\begin{aligned}
\int_{(\mathbf R^3)^2}\mathrm{St}[f,f]\ln f(t,x,v)dxdv=\int_{(\mathbf R^3)^4}(-S_{11}+S_{12}+S_2-S_3-S_{41}-S_{42})
\\
\times f(t,x,v)f(t,y,w)b(|x-y|)(v-w|n_{y,x})_+dxdvdydw&\,.
\end{aligned}
\]

Set
\[
\Lambda[\psi](t):=\int_{(\mathbf R^3)^4}S_{11}f(t,x,v)f(t,y,w)b(|x-y|)(v-w|n_{y,x})_+dxdvdydw\,;
\]
clearly
\[
\psi\ge 0\implies\Lambda[\psi]\ge 0\,.
\]
\underline{Step 3.} Next
\[
\begin{aligned}
\Gamma_{12}:=\int_{(\mathbf R^3)^4}S_{12}f(t,x,v)f(t,y,w)b(|x-y|)(v-w|n_{y,x})_+dxdvdydw
\\
=\tfrac12\int_{(\mathbf R^3)^4}\psi(x,v'(v,w,n_{y,x}))(f(t,x,v'(v,w,n_{y,x}))f(t,y,w'(v,w,n_{y,x}))
\\
-f(t,x,v)f(t,y,w))b(|x-y|)(v-w|n_{y,x})_+dxdvdydw
\\
=\tfrac12\int_{(\mathbf R^3)^4}(\psi(t,x,v)(v-w|n_{x,y})_+-\psi(x,v'(v,w,n_{y,x}))(v-w|n_{y,x})_+)
\\
\times f(t,x,v)f(t,y,w)b(|x-y|)dxdvdydw
\end{aligned}
\]
where the change of variables $(v,w)\mapsto(v'(v,w,n_{y,x}),w'(v,w,n_{y,x}))$ leads to the last equality, upon observing that
\[
(v-w|n_{y,x})=(v'(v,w,n_{y,x})-w'(v,w,n_{y,x})|n_{x,y})\,.
\]
Since
\[
\psi(t,x,v)(v-w|n_{x,y})_+=\psi(t,x,v)(v-w|n_{x,y})+\psi(t,x,v)(v-w|n_{y,x})_+\,,
\]
the term $\Gamma_{12}$ is recast as
\[
\begin{aligned}
\Gamma_{12}=\tfrac12\int_{(\mathbf R^3)^4}\psi(x,v)b(|x-y|)(v-w|n_{x,y})f(t,x,v)f(t,y,w)dxdvdydw
\\
+\tfrac12\int_{(\mathbf R^3)^4}(\psi(x,v)-\psi(x,v'(v,w,n_{y,x}))(v-w|n_{y,x})_+
\\
\times f(t,x,v)f(t,y,w)b(|x-y|)dxdvdydw
\\
=\Gamma_{121}+\Gamma_{122}&\,.
\end{aligned}
\]
\underline{Step 4.} The term $\Gamma_{121}$ is handled exactly as the potential term in the global entropy inequality: 
using the identity \eqref{PotDer} leads to
\[
\begin{aligned}
\Gamma_{121}=-\tfrac12\int_{(\mathbf R^3)^4}\psi(x,v)f(t,x,v)f(t,y,w)(v\cdot\nabla_x+w\cdot\nabla_y)\beta(|x-y|)dxdvdydw
\\
=\tfrac12\int_{(\mathbf R^3)^4}(v\cdot\nabla_x+w\cdot\nabla_y)[\psi(x,v)f(t,x,v)f(t,y,w)]\beta(|x-y|)dxdvdydw
\\
=\tfrac12\int_{(\mathbf R^3)^4}v\cdot\nabla_x\psi(x,v)f(t,x,v)f(t,y,w)\beta(|x-y|)dxdvdydw
\\
+\tfrac12\int_{(\mathbf R^3)^4}\psi(x,v)\mathrm{St}[f,f](t,x,v)f(t,y,w)\beta(|x-y|)dxdvdydw
\\
-\tfrac12\frac{d}{dt}\int_{(\mathbf R^3)^4}\psi(x,v)f(t,x,v)f(t,y,w)\beta(|x-y|)dxdvdydw&\,.
\end{aligned}
\]
In the last equality above, we have used the equality
\[
\begin{aligned}
(v\cdot\nabla_x+w\cdot\nabla_y)[\psi(x,v)f(t,x,v)f(t,y,w)]=f(t,x,v)f(t,y,w)v\cdot\nabla_x\psi(x,v)
\\
+\psi(x,v)f(t,y,w)(\mathrm{St}[f,f](t,x,v)-\partial_tf(t,x,v))
\\
+\psi(x,v)f(t,x,v)(\mathrm{St}[f,f](t,y,w)-\partial_tf(t,y,w))
\\
=f(t,x,v)f(t,y,w)v\cdot\nabla_x\psi(x,v)
\\
+\psi(x,v)(f(t,y,w)\mathrm{St}[f,f](t,x,v)+f(t,x,v)\mathrm{St}[f,f](t,y,w))
\\
-\psi(x,v)\partial_t(f(t,x,v)f(t,y,w))&\,,
\end{aligned}
\]
and the local conservation of mass (the first identity in Proposition \ref{P-GlobCons})
\[
\int_{\mathbf R^3}\mathrm{St}[f,f](t,y,w)dw=0
\]
to dispose of the term $\psi(x,v)f(t,x,v)\mathrm{St}[f,f](t,y,w)$ in the integral.

Representing $\mathrm{St}[f,f](t,x,v)$ in terms of the Landau mass current $\mathcal J_0[f]=\mathcal J_0[f,f]$ (see \eqref{DivMassCurrent}, and the definition 
of $\mathcal J_0[f,g]$ in the statement of Theorem \ref{T-HThmLoc}),
\[
\begin{aligned}
\Gamma_{121}=\tfrac12\int_{(\mathbf R^3)^2}v\cdot\nabla_x\psi(x,v)f(t,x,v)\rho_f(t,\cdot)\star\beta(|\cdot|)(x)dxdv
\\
-\tfrac12\int_{(\mathbf R^3)^2}\nabla_v\psi(x,v)\cdot\mathcal J_0[f](t,x,v)\rho_f(t,\cdot)\star\beta(|\cdot|)(x)dxdv
\\
-\tfrac12\frac{d}{dt}\int_{(\mathbf R^3)^2}\psi(x,v)f(t,x,v)\rho_f(t,\cdot)\star\beta(|\cdot|)(x)dxdv&\,.
\end{aligned}
\]

 \underline{Step 5.} On the other hand
\[
\begin{aligned}
\Gamma_{122}=&\tfrac12\int_{(\mathbf R^3)^4}(\psi(x,v)-\psi(x,v'(v,w,n_{y,x}))(v-w|n_{y,x})_+
\\
&\hskip 1.2cm\times f(t,x,v)f(t,y,w)b(|x-y|)dxdvdydw
\\
=&-\tfrac12\int_{(\mathbf R^3)^2}\psi(x,v)\mathrm{St}[f,f](t,x,v)dxdv
\\
=&\tfrac12\int_{(\mathbf R^3)^2}\nabla_v\psi(x,v)\cdot\mathcal J_0[f](t,x,v)dxdv\,.
\end{aligned}
\]
Next observe that
\[
S_3+S_{42}=\tfrac12(\psi(x,v)-\psi(x,v'(v,w,n_{y,x}))\ln(f(t,x,v)f(t,y,w))\,,
\]
so that
\[
\begin{aligned}
\Gamma_{34}:=\int_{(\mathbf R^3)^4}(S_3+S_{42})f(t,x,v)f(t,y,w)b(|x-y|)(v-w|n_{y,x})_+dxdvdydw
\\
=\tfrac12\int_{(\mathbf R^3)^4}(\psi(x,v)-\psi(x,v'(v,w,n_{y,x}))f(t,x,v)f(t,y,w)\ln(f(t,x,v)f(t,y,w))
\\
\times b(|x-y|)(v-w|n_{y,x})_+dxdvdydw
\\
=-\tfrac12\int_{(\mathbf R^3)^2}\nabla_v\psi(x,v)\cdot(\mathcal J_0[f\ln f,f](t,x,v)+\mathcal J_0[f,f\ln f](t,x,v))dxdv&\,.
\end{aligned}
\]
 \underline{Step 6.} While the terms $\Gamma_{12}$ and $\Gamma_{34}$ only contribute divergences in $v$ and time derivatives, the two terms considered below contribute 
 divergences in $x$.
 
 First consider
\[
\begin{aligned}
\Gamma_{41}:=\int_{(\mathbf R^3)^4}S_{41}f(t,x,v)f(t,y,w)b(|x-y|)(v-w|n_{y,x})_+dxdvdydw
\\
=\!\tfrac12\!\int_{(\mathbf R^3)^4}\!\!(\psi(y,w)\!-\!\psi(x,v))f(t,x,v)f\ln f(t,y,w)b(|x\!-\!y|)(v\!-\!w|n_{y,x})_+dxdvdydw&\,,
\end{aligned}
\]
and express 
\[
\psi(y,w)-\psi(x,v)=\int_0^{\pi/2}\frac{d}{d\theta}\psi(x_\theta,v_\theta)d\theta\,,
\]
where
\[
\left(\begin{matrix}x_\theta\\ y_\theta\end{matrix}\right):=\left(\begin{matrix}\cos\theta\,\,\,&\sin\theta\\ -\sin\theta&\cos\theta\end{matrix}\right)\left(\begin{matrix}x\\ y\end{matrix}\right)
\qquad\text{ and }\qquad
\left(\begin{matrix}v_\theta\\ w_\theta\end{matrix}\right):=\left(\begin{matrix}\cos\theta\,\,\,&\sin\theta\\ -\sin\theta&\cos\theta\end{matrix}\right)\left(\begin{matrix}v\\ w\end{matrix}\right)\,.
\]
Since
\[
\frac{d}{d\theta}\psi(x_\theta,v_\theta)=y_\theta\cdot\nabla_x\psi(x_\theta,v_\theta)+w_\theta\cdot\nabla_v\psi(x_\theta,v_\theta)\,,
\]
one has
\[
\begin{aligned}
\Gamma_{41}=\tfrac12\int_{(\mathbf R^3)^4}(y_\theta\cdot\nabla_x\psi(x_\theta,v_\theta)+w_\theta\cdot\nabla_v\psi(x_\theta,v_\theta))f(t,x,v)f\ln f(t,y,w)
\\
\times b(|x-y|)(v-w|n_{y,x})_+dxdvdydw
\\
=\tfrac12\int_{(\mathbf R^3)^4}(y\cdot\nabla_x\psi+w\cdot\nabla_v\psi)(x,v)\int_0^{\pi/2}f(t,x_{-\theta},v_{-\theta})f\ln f(t,y_{-\theta},w_{-\theta})
\\
\times b(|x_{-\theta}-y_{-\theta}|)(v_{-\theta}-w_{-\theta}|n_{y_{-\theta},x_{-\theta}})_+d\theta dxdvdydw
\end{aligned}
\]
because
\[
\det\frac{\partial(x_\theta,y_\theta)}{\partial(x,y)}=\det\frac{\partial(v_\theta,w_\theta)}{\partial(v,w)}=1\,.
\]
 \underline{Step 7.} Similarly, set
\[
\begin{aligned}
\Gamma_2:=\int_{(\mathbf R^3)^4}S_2f(t,x,v)f(t,y,w)b(|x-y|)(v-w|n_{y,x})_+dxdvdydw
\\
=\tfrac12\int_{(\mathbf R^3)^4}(\psi(y,w)-\psi(x,v))f(t,x,v'(v,w,n_{y,x}))\ln f(t,y,w)
\\
\times f(t,y,w'(v,w,n_{y,x}))b(|x-y|)(v-w|n_{x,y})_+dxdvdydw&\,.
\end{aligned}
\]
By the same change of variables as in $\Gamma_2$, one finds that
\[
\begin{aligned}
\Gamma_2=\tfrac12\int_{(\mathbf R^3)^4}(y\cdot\nabla_x\psi+w\cdot\nabla_v\psi)(x,v)\int_0^{\pi/2}f(t,x_{-\theta},v'(v,w,n_{y,x})_{-\theta})
\\
\times f(t,y_{-\theta},w'(v,w,n_{y,x})_{-\theta})\ln f(t,y_{-\theta},w_{-\theta})b(|x_{-\theta}-y_{-\theta}|)
\\
\times (v_{-\theta}-w_{-\theta}|n_{x_{-\theta},y_{-\theta}})_+d\theta dxdvdydw&\,.
\end{aligned}
\]
 \underline{Conclusion.} Gathering all these terms together, we arrive at the following identity:
 \[
 \begin{aligned}
 \int_{(\mathbf R^3)^2}\mathrm{St}[f,f](t,x,v)\ln f(t,x,v)\psi(x,v)dxdv=-\Lambda[\psi](t)
 \\
 +\int_{(\mathbf R^3)^2}(\mathcal K[f](t,x,v)\cdot\nabla_x\psi(x,v)+\mathcal L[f](t,x,v)\cdot\nabla_v\psi(x,v))dxdv
 \\
 -\tfrac12\frac{d}{dt}\int_{(\mathbf R^3)^2}\psi(x,v)f(t,x,v)\rho_f(t,\cdot)\star\beta(|\cdot|)(x)dxdv&\,,
 \end{aligned}
\]
where
\[
\begin{aligned}
\mathcal L[f](t,x,v):=-\tfrac12\mathcal J_0[f,f](t,x,v)(1-\rho_f\star\beta(|\cdot|)(t,x))
\\
-\tfrac12(\mathcal J_0[f\ln f,f](t,x,v)+\mathcal J_0[f,f\ln f](t,x,v))
\\
+\tfrac12\int_{(\mathbf R^3)^2}\int_0^{\pi/2}(f(t,x_{-\theta},v'(v,w,n_{y,x})_{-\theta})f(t,y_{-\theta},w'(v,w,n_{y,x})_{-\theta})
\\
-f(t,x_{-\theta},v_{-\theta})f(t,y_{-\theta},w_{-\theta}))\ln f(t,y_{-\theta},w_{-\theta})b(|x_{-\theta}-y_{-\theta}|)
\\
\times (v_{-\theta}-w_{-\theta}|n_{x_{-\theta},y_{-\theta}})_+wd\theta dydw&\,,
\end{aligned}
\]
where the notation $\rho_f(t,\cdot)\star\beta(|\cdot|)$ has been defined in \eqref{DefConvol}, while
\[
\begin{aligned}
\mathcal K[f](t,x,v)\!:=\!\tfrac12\int_{(\mathbf R^3)^2}\!\int_0^{\pi/2}\!(f(t,x_{-\theta},v'(v,w,n_{y,x})_{-\theta})f(t,y_{-\theta},w'(v,w,n_{y,x})_{-\theta})
\\
-f(t,x_{-\theta},v_{-\theta})f(t,y_{-\theta},w_{-\theta}))\ln f(t,y_{-\theta},w_{-\theta})b(|x_{-\theta}-y_{-\theta}|)
\\
\times (v_{-\theta}-w_{-\theta}|n_{x_{-\theta},y_{-\theta}})_+yd\theta dydw&\,.
\end{aligned}
\]


\section{Conclusion}


Summarizing, we have seen that delocalized collision integrals in the style of the Enskog, Povzner or soft sphere integrals can be expressed in terms of momentum and energy currents which have
nontrivial components in the space of positions. For that reason, such delocalized collision integrals contribute to the local conservation laws of momentum and energy which appear in the Euler
system of gas dynamics in the case of dense gases. This effect has been analyzed at the level of asymptotic analysis in \cite{Lachowicz} by means of the Hilbert expansion method, along with the
resulting virial expansion of the pressure. One can also find in chapter 16 of \cite{CC} (see in particular \S 16.32) a variant of the discussions in the present paper limited to an asymptotic expansion
in terms of the molecular radius. What is proposed in the present article is a non-asymptotic discussion which parallels that of \S 16.32 in \cite{CC}, along with a discussion of a local variant of the
Boltzmann H Theorem in Theorem \ref{T-HThmLoc}. All our discussions are based on Landau's formulation of the Boltzmann collision integral in terms of a mass current in the space of molecular
velocities, to be found in \S 41 of \cite{LL10}, and discussed in \cite{VillaniLandauCurr} in full detail.


\end{document}